\documentclass[twoside,english,onefignum,onetabnum]{siamart190516}

\usepackage{amsmath}
\usepackage{nicefrac}
\usepackage[version=4]{mhchem}
\usepackage{bm}
\usepackage{blindtext}
\newsiamthm{property}{Property}

\makeatletter
\newcommand*\bigcdot{\mathpalette\bigcdot@{.7}}
\newcommand*\bigcdot@[2]{\mathbin{\vcenter{\hbox{\scalebox{#2}{$\m@th#1\bullet$}}}}}
\makeatother

\usepackage{arydshln}
\usepackage{amsfonts} 

\newsiamthm{claim}{Claim}

\headers{SVD of stoichiometry matrix}{J.M. Wentz and D.M. Bortz}
\title{Analytical singular value decomposition for a class of stoichiometry matrices\thanks{
\funding{This work was supported in part by an NSF GRFP and in part by NSF IGERT Grant Number 1144807.}}}
\author{Jacqueline M. Wentz\thanks{Department of Applied Mathematics, University of Colorado, Boulder, CO
		(\email{jacqueline.wentz@colorado.edu}, \email{dmbortz@colorado.edu}).}
	\and David M. Bortz\footnotemark[2]}
	
\begin{document}
\maketitle
\begin{abstract}
	We present the analytical singular value decomposition of the stoichiometry matrix for a spatially discrete reaction-diffusion system on a one dimensional domain. The domain has two subregions which share a single common boundary. Each of the subregions is further partitioned into a finite number of compartments. Chemical reactions can occur within a compartment, whereas diffusion is represented as movement between adjacent compartments. Inspired by biology, we study both 1) the case where the reactions on each side of the boundary are different and only certain species diffuse across the boundary as well as 2) the case with spatially homogenous reactions and diffusion. We write the stoichiometry matrix for these two classes of systems using a Kronecker product formulation. For the first scenario, we apply linear perturbation theory to derive an approximate singular value decomposition in the limit as diffusion becomes much faster than reactions. For the second scenario, we derive an exact analytical singular value decomposition for all relative diffusion and reaction time scales. By writing the stoichiometry matrix using Kronecker products, we show that the singular vectors and values can also be written concisely using Kronecker products.

   Ultimately, we find that the singular value decomposition of the reaction-diffusion stoichiometry matrix depends on the singular value decompositions of smaller matrices. These smaller matrices represent modified versions of the reaction-only stoichiometry matrices and the analytically known diffusion-only stoichiometry matrix. Our results provide a mathematical framework that can be used to study complex biochemical systems with metabolic compartments.  MATLAB code for calculating the SVD equations is available at \url{www.github.com/MathBioCU/ReacDiffStoicSVD}.
\end{abstract}

\section{Introduction}\label{sec:intro}

In stoichiometric network analysis the mass balance equation for a reaction-only system is written as follows
\begin{equation}
\frac{dw}{dt} = S_rf
\label{eq:MassBalance}
\end{equation}
where $w$ is a species concentration vector, $S_r$ is the stoichiometry matrix, and $f$ is a vector of reaction fluxes \cite{Clarke1988}. We use the subscript $r$ to refer to a stoichiometry matrix that only describes reactive processes. Although the flux vector $f$ is a function of the species concentration, the formulation given by (\ref{eq:MassBalance}) avoids assumptions about the form of the kinetic equations that relate the fluxes to the species concentration (e.g., mass-action \cite{Voit2015} or Michaelis-Menton kinetics \cite{Johnson2011}).  The stoichiometry matrix contains information about the species involved in each reaction. As a simple example, consider the following set of reactions:
\begin{equation}
   \O\ce{ -> A}, \quad\quad \ce{A -> B}, \quad\quad \ce{B ->}\O.
\end{equation}
 Here, species $A$ is produced, transitions into species $B$, and species $B$ decays. The stoichiometry matrix for this example system is
\begin{align}\label{eq:Sr}
   S_r &=
      \begin{bmatrix}
         1 & -1 & 0 \\
         0 & 1 & -1
      \end{bmatrix}.
\end{align}
The first row of $S_r$ corresponds to species $A$ and the second row corresponds to species $B$. Each of the three columns correspond to the three reactions, respectively.

The analysis of $S_r$ provides information on structural properties of the system without requiring kinetic information. In particular, the singular value decomposition (SVD) of $S_r$ provides information on systemic properties, including decoupled eigenreactions (i.e., linear combinations of species that are moved by linear combinations of fluxes), conservation relations, and fluxes that can exist in the system under steady-state conditions \cite{Palsson2006}. This type of analysis can be used to determine hidden relationships in a network and compare biochemical properties amongst different organisms \cite{Famili2003,Palese2012}.

Here, our goal is to derive the SVD of a stoichiometry matrix that, in addition to the reactions, includes information on the spatial properties of a system. Specifically, we define the stoichiometry matrix for a one dimensional spatially discrete system by considering both the reactions in each spatial compartment as well as the movement of species between adjacent compartments. We refer to this matrix as the reaction-diffusion (RD) stoichiometry matrix and write it using the reaction-only and diffusion-only stoichiometry matrices, i.e. $S_r$ and $S_d$, respectively. The reaction only stoichiometry matrix is as described in (\ref{eq:Sr}) and the diffusion only stoichiometry matrix can be thought of as representing a single species diffusing through space. As an example, in a system with one diffusing species, three spatial compartments, and homogeneous Neumann boundary conditions, the diffusion-only stoichiometry matrix is
\begin{equation}
   S_d =
      \begin{bmatrix}
         0 & -1 & 0 & 0 \\
         0 & 1 & -1 & 0 \\
         0 & 0 & 1 & 0
      \end{bmatrix}.
\end{equation}
Analagously to $S_r$, each row in $S_d$ corresponds to the species in each of the three compartments. The first and last column of $S_d$ correspond to species movement across the boundary of the domain and, for this example, contain only zeros due to the homogenous Neumann boundary conditions. The middle columns represent the movement of the species between adjacent compartments.

Using the reaction-only and diffusion-only stoichiometry matrix definitions, in a system with both reactions and diffusion where there are $n$ spatial compartments and $m$ species that freely diffuse through space, the stoichiometry matrix $S$ is
\begin{equation}\label{eq:intro}
	S =
	\begin{bmatrix}
		\gamma S_r \otimes I_n & I_m \otimes S_d
	\end{bmatrix},
\end{equation}
where $\otimes$ represents the Kronecker product \cite{Loan2000}, $\gamma>0$ describes the relative rate of reactions to diffusion, and $I_{a}$ is the identity matrix of size $a$. Here, the $S_r \otimes I_n$ block represents the reactions occurring in each compartment, whereas the $I_m \otimes S_d$ block represents diffusive movement. The Kronecker product has previously been used to compactly represent diffusion for the spatially discrete reaction-diffusion ODE system \cite{Arcak2011,DAutilia2020}. We previously developed criteria to guarantee a version of this ODE system is bounded for all time \cite{Wentz2020}. Here, we instead use (\ref{eq:intro}) to study the spatially discrete system in the context of stoichiometric network analysis. We write both the reactive and diffusive terms using a Kronecker product formulation as this will simplify the SVD derivation.

In this paper we will consider a more general form of (\ref{eq:intro}) where, in addition to diffusion, there is a spatial barrier in the system that divides the 1D domain into two subregions. We consider this class of systems because it allows our results to be applied to study, for example, the effect of metabolic compartmentalization within a cell. We will use concepts from linear perturbation theory \cite{Kato1995} to derive the approximate SVD in the limit as diffusion becomes much faster than reactions. We additionally consider the special case where diffusion of all species is allowed freely throughout the domain, i.e., where the stoichiometry matrix can be written  as given by (\ref{eq:intro}). We show that, for this scenario the SVD becomes exact for all values of $\gamma$. The derived SVDs for the system with and without a spatial barrier depend on the SVDs of smaller matrices, such as the reaction-only stoichiometry matrix.

To help provide structure and guide our argument, in Section~\ref{sec:main-results} we chose to present the main result first (see Theorems~\ref{thm:svd-S} and \ref{thm:svd-S-II}). We then provide a more complete set of definitions and notation in Section~\ref{sec:notation}. This includes a complete description of the system as well as definitions of matrices whose SVDs are used to write main result. In Section~\ref{sec:useful-results} we provide preliminary results that will be helpful for proving Theorem~\ref{thm:svd-S}. In Section~\ref{sec:SVD} we provide the complete proofs of Theorem~\ref{thm:svd-S} and Theorem~\ref{thm:svd-S-II}. Finally, in Section~\ref{sec:discussion} we provide some intuition for the SVD equations and discuss potential applications of this work.

\section{System description and statement of main result}
\label{sec:main-results}

Here we provide a brief description of the system and state the main result. For a thorough description of the notation and definitions used see Section~\ref{sec:notation}.

We consider a one dimensional, spatially-discrete, reaction diffusion system that is divided into two subregions. A subset of the species is allowed to diffuse between the two subregions, and we allow for different sets of reactions to occur in each region. We will consider three boundary conditions: no input/output fluxes, input/output fluxes at one boundary point, and input/output fluxes at both boundary points. As an example, biologically this system description might represent a radially symmetric cell, where the two subregions are the cytoplasm and the nucleus.

The stoichiometry matrix for this class of systems can be written as
\begin{equation}\label{eq:S}
   S :=
      \begin{bmatrix}
         \gamma
         \begin{bmatrix}
            I_{n_1}\otimes S_{r_1} & 0 \\
            0 & I_{n_2} \otimes S_{r_2}
         \end{bmatrix}& S_{d}\otimes D_+ + (S_d - H)\otimes D_{-}
      \end{bmatrix}
\end{equation}
where the first column block represents reactive processes and the second represents diffusive processes. Here, $S_{r_1}$ and $S_{r_2}$ represent the reaction-only stoichiometry matrices for each of the two subregions, $S_d \otimes D_+$ describes the diffusion of species that move across the entire domain (i.e., species that can cross the barrier between the two subregions), and $(S_d - H) \otimes D_-$ describes the diffusion of species that stay within a single subregion. The parameter $\gamma > 0$ represents the relative rate of reactions compared with diffusion.

In this section we present the SVD of the stoichiometry matrix given by (\ref{eq:S}) in the limit as diffusion becomes much faster than reactions, i.e., as $\gamma \rightarrow 0$. Briefly, the main result depends on the SVD of smaller reaction-only and diffusion-only systems. This includes matrices that only involve reactive processes, which will be written using variations of $S_r$ (e.g., $S_{r_1}$, $S_{r_2}$), and matrices that only involve diffusive processes, which will be written using variations of $S_d$ (e.g., $S_{d_1}$, $S_{d_2}$).

Our general notation for writing down the SVD of $S_{\bigcdot} \in \mathbb{R}^{s_1 \times s_2}$ will be as follows:
\begin{equation}
   S_{\bigcdot} = U_{\bigcdot}\Sigma_{\bigcdot}V_{\bigcdot}^{T}=
      \begin{bmatrix}
      \hat{U}_{\bigcdot} & \breve{U}_{\bigcdot}
      \end{bmatrix}
      \begin{bmatrix}
      \hat{\Sigma}_{\bigcdot} & 0 \\
      0 & 0
      \end{bmatrix}
      \begin{bmatrix}
      \hat{V}_{\bigcdot} & \breve{V}_{\bigcdot}
      \end{bmatrix}^{T}.
   \label{eq:svd-Sstar}
\end{equation}
We will refer to the rank of $S_{\bigcdot}$ as $q_{\bigcdot}$ and the size of the nullspace as $\breve{q}_{\bigcdot}$. In some cases the singular vectors will be divided into two components (e.g., $U_d = [U_{d,a};U_{d,b}]$). With this SVD notation in mind, we next state the main result of the paper. Although the complete definitions and notations are not given until Section~\ref{sec:notation}, it is possible to immediately see that the SVD depends only on SVDs of variations of stoichiometry matrices for the reaction-only and diffusion-only systems.

\begin{theorem}\label{thm:svd-S}
	As $\gamma \rightarrow 0$ the unsorted SVD of $S$, as given by (\ref{eq:S}), is
	\begin{equation}
		S = U \Sigma V^T=
		\begin{bmatrix}
			\hat{U} & \breve{U}
		\end{bmatrix}
		\begin{bmatrix}
			\hat{\Sigma} & 0 \\
			0 & 0
		\end{bmatrix}
		\begin{bmatrix}
			\hat{V} & \breve{V}
		\end{bmatrix}^{T}.
	\end{equation}
	where the singular vectors that have nonzero singular values are given by six components, $\hat{U}_i$, $\hat{V}_i$, such that
   \begin{equation}
      S = \sum_{i\in\hat{\mathcal{J}}} \hat{U}_{1,j} \hat{\Sigma}_{1,j} \hat{V}_{1,j}^T + \sum_{i=2}^6 \hat{U}_i \hat{\Sigma}_i \hat{V}_i^T
   \end{equation}
where
	\begin{equation*}
      \begin{aligned}
         \hat{U}_{1,j} &= u_d^{(j)} \otimes I_m^{\mathcal{M}_+} U_{\bar{r}_+,j} &
         \hat{V}_{1,j} &= \left[
            \begin{array}{c}
               \frac{\gamma}{|u_{d,s_1}^{(j)}|}  u_{d,s_1}^{(j)} \otimes  V_{\bar{r}_+,j,s_1} \Sigma_{\bar{r}_+,j}^T \\
               \frac{\gamma}{|u_{d,s_2}^{(j)}|} u_{d,s_2}^{(j)} \otimes  V_{\bar{r}_+,j,s_2} \Sigma_{\bar{r}_+,j}^T
               \vspace{3pt}\\
               \hdashline \noalign{\vskip 3pt}
               v_d^{(j)} \sigma_d^{(j)} \otimes
               I_m^{\mathcal{M}_+} U_{\bar{r}_+,j}
   			\end{array}\right] \hat{\Sigma}_{1,j}^{-1} \\
         \hat{U}_2 &=
      		\begin{bmatrix}
      			U_{d_1}^{\hat{\mathcal{J}}_1} \\
      			0
      		\end{bmatrix} \otimes I_m^\mathcal{M_-} U_{r_1,-} &
         \hat{V}_2 &=\left[
            \begin{array}{c}
   				\gamma U_{d_1}^{\hat{\mathcal{J}}_1} \otimes V_{r_1,-} \Sigma_{r_1,-}^T \\
   				0 \\
   				\hdashline \noalign{\vskip 3pt}
   				\left(V_{d_1} \Sigma_{d_1}\right)^{\hat{\mathcal{J}}_1} \otimes I_m^{\mathcal{M}_-} U_{r_1,-} \\
   				0
   			\end{array}\right] \hat{\Sigma}_2^{-1}  \\
         \hat{U}_3 &=
      		\begin{bmatrix}
      			U_{d_1}^{n_1 \setminus (\mathcal{J}_1^C \cup \hat{\mathcal{J}}_1)}  \\
      			0
      		\end{bmatrix} \otimes I_m^\mathcal{M_-} \hat{U}_{r_1,-} &
         \hat{V}_3 &=\left[
            \begin{array}{c}
   				\gamma U_{d_1}^{n_1 \setminus (\mathcal{J}_1^C \cup \hat{\mathcal{J}}_1)} \otimes \hat{V}_{r_1,-} \hat{\Sigma}_{r_1,-}^T \\
   				0 \\
   				\hdashline \noalign{\vskip 3pt}
   				0
   			\end{array}\right] \hat{\Sigma}_3^{-1}  \\
         \hat{U}_4 &=
      		\begin{bmatrix}
      			0 \\
      			U_{d_2}^{\hat{\mathcal{J}}_2}
      		\end{bmatrix}
      		\otimes I_m^{\mathcal{M}_-}U_{r_2,-} &
         \hat{V}_4 &=
   			\left[\begin{array}{c}
   				0 \\
   				\gamma U_{d_2}^{\hat{\mathcal{J}}_2}
   				\otimes
   				V_{r_2,-} \Sigma^T_{r_2,-}
   				\vspace{3pt} \\
   				\hdashline \noalign{\vskip 2pt}
   				0 \\
   				\left(V_{d_2} \Sigma_{d_2}\right)^{\hat{\mathcal{J}}_2}
   				\otimes I_m^{\mathcal{M}_-} U_{r_2,-}
   			\end{array}\right] \hat{\Sigma}_4^{-1} \\
         \hat{U}_5 &=
      		\begin{bmatrix}
      			\frac{1}{C_1} U_{d,s_1}^{\hat{\mathcal{J}}^C} \otimes U_{\bar{r},m_1} \\
      			\frac{1}{C_2} U_{d,s_2}^{\hat{\mathcal{J}}^C} \otimes U_{\bar{r},m_2}
      		\end{bmatrix} &
         \hat{V}_5 &=
   			\left[\begin{array}{cc} \vspace{5pt}
   				\frac{\gamma}{C_1} U_{d,s_1}^{\hat{\mathcal{J}}^C} \otimes
   				\begin{bmatrix}
   					\hat{V}_{\bar{r},s_1} \hat{\Sigma}_{\bar{r}} & 0
   				\end{bmatrix} \\
   				\frac{\gamma}{C_2} U_{d,s_2}^{\hat{\mathcal{J}}^C} \otimes
   				\begin{bmatrix}
   					\hat{V}_{\bar{r},s_2} \hat{\Sigma}_{\bar{r}} & 0
   				\end{bmatrix} \vspace{3pt} \\
   				\hdashline \noalign{\vskip 3pt}
   				\frac{1}{C_1}V_{d,s_1}^{\hat{\mathcal{J}}^C} \Sigma_d^{\hat{\mathcal{J}}^C} \otimes
   				U_{\bar{r},m_1} \\
   				\frac{1}{C_2}V_{d,s_2}^{\hat{\mathcal{J}}^C} \Sigma_d^{\hat{\mathcal{J}}^C} \otimes
   				U_{\bar{r},m_2}
   			\end{array}\right] \hat{\Sigma}_5^{-1}\\
         \hat{U}_6 &=
      		\begin{bmatrix}
      			\frac{1}{C_1} \breve{U}_{d,s_1} \otimes \hat{U}_{\bar{r},m_1} \\
      			\frac{1}{C_2} \breve{U}_{d,s_2} \otimes \hat{U}_{\bar{r},m_2}
      		\end{bmatrix} &
         \hat{V}_6 &=
   			\left[\begin{array}{cc} \vspace{5pt}
   				\frac{\gamma}{C_1} \breve{U}_{d,s_1} \otimes \hat{V}_{\bar{r},s_1} \hat{\Sigma}_{\bar{r}} \\
   				\frac{\gamma}{C_2} \breve{U}_{d,s_2} \otimes \hat{V}_{\bar{r},s_2} \hat{\Sigma}_{\bar{r}} \\
   				\hdashline \noalign{\vskip 3pt}
   				0
   			\end{array}\right] \hat{\Sigma}_6^{-1}
		\end{aligned}
	\end{equation*}
	and
	\begin{equation*}
   \begin{aligned}
		\hat{\Sigma}_{1,j}^2 &= \left(\sigma_d^{(j)}\right)^2 I_{m_+} +		\gamma^2 \left(\Sigma_{\bar{r}_+,j}^2\right)_{m_+} &
		\hat{\Sigma}_2^2 &= (\hat{\Sigma}_{d_1}^2)^{\hat{\mathcal{J}}_1} \oplus \gamma^2 \left(\hat{\Sigma}_{r_1,-}^2\right)_{m_-} \\
		\hat{\Sigma}_3^2 &= (\Sigma_{d_1}^2)^{n_1\setminus (\mathcal{J}_1^C \cup \hat{\mathcal{J}}_1)} \oplus \gamma^2 \hat{\Sigma}_{r_1,-}^2 &
		\hat{\Sigma}_4^2 &= (\Sigma_{d_2}^2)^{\hat{\mathcal{J}}_2} \oplus \gamma^2 \left(\hat{\Sigma}_{r_2,-}^2\right)_{m_-} \\
		\hat{\Sigma}_5^2 &=
         (\hat{\Sigma}_d^2)^{\hat{\mathcal{J}}^C} \oplus
            \gamma^2\left(\hat{\Sigma}_{\bar{r}}^2\right)_{q_{\bar{r}}+\breve{q}_{\bar{r}}} &
      \hat{\Sigma}_6^2 &=
         I_{n-q_d} \otimes \gamma^2 \hat{\Sigma}_{\bar{r}}^2.
   \end{aligned}
	\end{equation*}
	A basis for the left nullspace of $S$ is
	\begin{equation}
		\breve{U} =
         \begin{bmatrix}
   			\breve{U}_1 & \breve{U}_2
   		\end{bmatrix}
	\end{equation}
	where
	\begin{align*}
		\breve{U}_1 &=
			\begin{bmatrix}
				U_{d_1}^{n_1\setminus (\mathcal{J}_1^C \cup \hat{\mathcal{J}}_1)} \\
				0
			\end{bmatrix} \otimes I_m^\mathcal{M_-} \breve{U}_{r_1,-} &
		\breve{U}_2 &=
			\begin{bmatrix}
				\frac{1}{C_1} \breve{U}_{d,s_1} \otimes \breve{U}_{\bar{r},s_1} \\
				\frac{1}{C_2} \breve{U}_{d,s_2} \otimes \breve{U}_{\bar{r},s_2}
			\end{bmatrix}
	\end{align*}
	and a basis for the (right) nullspace of $S$ is
	\begin{equation}
		\breve{V} = \begin{bmatrix}
			\breve{V}_1 & \breve{V}_2 &	\breve{V}_3 & \breve{V}_4 & \breve{V}_5
		\end{bmatrix}
	\end{equation}
	where
	\begin{equation}
		\begin{aligned}
		\breve{V}_1 &=
			\left[\begin{array}{cc}
				U_{d_1}^{n_1\setminus (\mathcal{J}_1^C \cup \hat{\mathcal{J}}_1)} \otimes \breve{V}_{r_1,-} \\
				0 \\
				\hdashline \noalign{\vskip 3pt}
				0
			\end{array}\right]  &
      \breve{V}_2 &=
			\left[\begin{array}{c}
				\hat{U}_{d_1} \hat{\Sigma}_{d_1} \otimes V_{r_1} \\
				0 \\
				\hdashline \noalign{\vskip 3pt}
				-\gamma \hat{V}_{d_1} \otimes U_{r_1} \Sigma_{r_1} \\
				0
			\end{array}\right] \left(\hat{\Sigma}_{d_1}^2 \oplus \gamma^2 \left(\hat{\Sigma}_{r_1}^2\right)_{p_1}\right)^{-\frac{1}{2}} \\
      \breve{V}_3 &=
			\left[\begin{array}{c} \vspace{5pt}
				\frac{1}{C_1} \breve{U}_{d,s_1} \otimes \breve{V}_{\bar{r},s_1} \\
				\frac{1}{C_2} \breve{U}_{d,s_2} \otimes \breve{V}_{\bar{r},s_2} \\
				\hdashline \noalign{\vskip 3pt}
				0
			\end{array}\right] &
		\breve{V}_4 &=
			\left[\begin{array}{c}
				0 \\
				\hat{U}_{d_2} \hat{\Sigma}_{d_2} \otimes V_{r_2} \vspace{3pt}\\
				\hdashline \noalign{\vskip 2pt}
				0 \\
				-\gamma \hat{V}_{d_2} \otimes U_{r_2} \Sigma_{r_2}
			\end{array}\right] \left(\hat{\Sigma}_{d_2}^2 \oplus \gamma^2 \left(\hat{\Sigma}_{r_2}^2\right)_{p_2}\right)^{-\frac{1}{2}} \\
	  \noalign{\noindent $\breve{V}_5 =
			\left[\begin{array}{cc}
				0 & 0\\
				\hdashline \noalign{\vskip 3pt}
				 \breve{V}_{d} \otimes I_m^{\mathcal{M}_+} & I_{n+1}^{\mathcal{B}} \otimes I_m^{\mathcal{M}_-}
			\end{array}\right]$.}
	\end{aligned}
\end{equation}
\end{theorem}
Note that the horizontal dashed lines used in the definition of the right singular vectors separate the vectors into components that correspond to the reactive fluxes (above dashed line) and diffusive fluxes (below dashed line). The proof of this theorem is given in Section \ref{sec:SVD}.

We have defined the SVD in Theorem~\ref{thm:svd-S} to be applicable for all three boundary conditions. Note that $\hat{U}_6$, $\hat{V}_6$, $\breve{U}_2$, and $\breve{V}_3$ are only nonempty for homogeneous Neumann boundary conditions (i.e., no input/output fluxes) and $\hat{U}_3$, $\hat{V}_3$, $\breve{U}_1$, and $\breve{V}_1$ are only nonempty when there is an input/output flux at a single boundary point.

The results given in Theorem~\ref{thm:svd-S} are simplified significantly when we consider systems that only have one region and spatially-homogeneous reactions. For such systems the stoichiometry matrix is simplified to
\begin{equation}\label{eq:S_II}
   S = \begin{bmatrix}
		\gamma I_n \otimes S_r & S_d \otimes I_m
	\end{bmatrix}.
\end{equation}
and the SVD is given by the following theorem.

\begin{theorem}\label{thm:svd-S-II}
	The SVD of the stoichiometry matrix $S$, as given by (\ref{eq:S_II}) is
	\begin{equation}
		S=U\Sigma V^{T}=
		\begin{bmatrix} \hat{U} & \breve{U} \end{bmatrix}
		\begin{bmatrix}
		\hat{\Sigma} & 0\\
		0 & 0
		\end{bmatrix}
		\begin{bmatrix} \hat{V} & \breve{V}\end{bmatrix}^{T}\label{eq:SVDFull}
	\end{equation}
	where
\begin{align*}
	\hat{U} &= \begin{bmatrix}\hat{U}_{d}\otimes \hat{U}_{r} & \hat{U}_{d}\otimes \breve{U}_{r} & \breve{U}_{d}\otimes \hat{U}_{r}\end{bmatrix} \\
	\breve{U} &= \breve{U}_{d}\otimes \breve{U}_{r} \\
	\hat{V} &=
	\begin{bmatrix}
		\gamma\left(\hat{U}_{d}\otimes \hat{V}_{r}\hat{\Sigma}_{r}\right)\tilde{\Sigma}^{-1} & 0 & \breve{U}_{d}\otimes \hat{V}_{r}\\
		\left(\hat{V}_{d}\hat{\Sigma}_{d}\otimes \hat{U}_{r}\right)\tilde{\Sigma}^{-1} & \hat{V}_{d}\otimes \breve{U}_{r} & 0
	\end{bmatrix} \\
	\breve{V} &=
	\begin{bmatrix}
		\left(\hat{U}_{d}\hat{\Sigma}_{d}\otimes \hat{V}_{r}\right)\tilde{\Sigma}^{-1} & \hat{U}_{d}\otimes \breve{V}_{r} & 0\\
		-\gamma\left(\hat{V}_{d}\otimes \hat{U}_{r}\hat{\Sigma}_{r}\right)\tilde{\Sigma}^{-1} & 0 & \breve{V}_{d}\otimes U_{r}
	\end{bmatrix}\\
	\hat{\Sigma} &=
	\begin{bmatrix}
		\tilde{\Sigma}\\
		& \hat{\Sigma}_{d}\otimes I_{m-q_{r}}\\
		&  & \gamma I_{n-q_{d}}\otimes\hat{\Sigma}_{r}
	\end{bmatrix} \\
	\tilde{\Sigma}^2 &= \hat{\Sigma}_{d}^{2}\oplus\left(\gamma\hat{\Sigma}_{r}\right)^{2}.
\end{align*}
\end{theorem}

\section{Notation and Definitions}\label{sec:notation}

Here we present notation and matrix definitions that are used to state and prove the main result. In Section~\ref{subsec:notation-matrix}, we provide basic notation for referring to matrices. In Section~\ref{subsec:notation-two-classes}, we present definitions used to define the discrete reaction-diffusion system. In Section~\ref{subsec:notation-sets} we define sets of indices that will be used for defining the SVD. In Section~\ref{sub-sec:notation-more-matrices}, we define the set of \textit{stoichiometry-like matrices} that are required for writing the SVD of the reaction-diffusion system. In Section \ref{subsec:notation-SVD}, we provide notation that, in addition to (\ref{eq:svd-Sstar}), will be used to define the SVD of relevant matrices. Table~\ref{tab} summarizes the notational defintions presented in this section.

\begin{table}
   \small
   \begin{tabular}{llp{7.8cm}}
      \textbf{Symbol} & \textbf{Size (if matrix)} & \textbf{Definition} \\ \hline
      $n$ & & Total number of spatial compartments \\
      $n_1$ & & Number of spatial compartments in Subregion 1 \\
      $n_2$ & & Number of spatial compartments in Subregion 2 \\
      $m$ & & Total number of species \\
      $p_1$ & & Number of reactions in Subregion 1 \\
      $p_2$ & & Number of reactions in Subregion 2 \\
      $C_1$, $C_2$ & & Constants dependent on the boundary conditions \\ \hline
      $\mathcal{M}_+$ & & Index set for species that diffuse across barrier \\
      $\mathcal{M}_-$ & & Index set for species that do not diffuse across barrier \\
      $\mathcal{B}$ & & Index set that depends on boundary conditions \\
      $\mathcal{J},\mathcal{J}_i$ & & Index sets of singular values that only occur in $\Sigma_{d}$, $\Sigma_{d_i}$ \\
      $\hat{\mathcal{J}},\hat{\mathcal{J}_i}$ & & Index sets of singular values that only occur in $\hat{\Sigma}_{d}$, $\hat{\Sigma}_{d_i}$ \\ \hline
      $S_d$ & $n\times n+1$ & Diffusion-only stoichiometry matrix for full domain\\
      $S_{d_1}$ & $n_1 \times n_1 + 1$ & Diffusion-only stoichiometry matrix for Subregion 1 \\
      $S_{d_2}$ & $n_2 \times n_2 + 1$ & Diffusion-only stoichiometry matrix for Subregion 2 \\
      $S_{r_1}$ & $m \times p_1$ & Reaction-only stoichiometry matrix for Subregion 1 \\
      $S_{r_2}$ & $m \times p_2$ & Reaction-only stoichiometry matrix for Subregion 2 \\
      $S_{r_i,+}$ & $m_+ \times p_i$ & Rows of $S_{r_i}$ for species that diffuse between subregions \\
      $S_{r_i,-}$ & $m_- \times p_i$ & Rows of $S_{r_i}$ for species that do not diffuse between subregions \\
      $S_{\bar{r}}$ & $2m \times p_1+p_2$ & Block matrix dependent on $S_{r_1}$ and $S_{r_2}$, see (\ref{eq:S_rbar}). \\
      $S_{\bar{r}_+,j}$ & $m_+ \times p_1+p_2$ & Block matrix dependent on $S_{r_1,+}$, $S_{r_2,+}$, and $u_d^{(j)}$, see (\ref{eq:S_rbar_+_j}).
   \end{tabular}
   \caption{Description of constants, index sets, and stoichiometry/stoichiometry-like matrices used to define the SVD of the RD stoichiometry matrix.}
   \label{tab}
\end{table}

\subsection{Matrix notation}\label{subsec:notation-matrix}

Matrices will be defined using uppercase letters (e.g., $A$) and sets of indices will be defined using calligraphic fonts (e.g., $\mathcal{B}$). We will use $A^{\mathcal{B}}$ to represent the columns of $A$ whose indices are in the set $\mathcal{B}$. If we refer to one column of a matrix (i.e., a column vector), we will typically use the lowercase letter and a superscript to refer to this column (i.e., the $i$th column of $A$ will be written as $a^{(i)}$). One exception to these rules will be for any diagonal matrix of singular values $\Sigma_{\bigcdot}$ and variations of this matrix. In this case, $\Sigma_{\bigcdot}^{\mathcal{B}}$ will represent a square diagonal matrix containing the singular values whose indices are in $\mathcal{B}$. Additionally $(\Sigma_{\bigcdot})_n$ will represent the matrix $\Sigma$ padded by zeros to make it size $n\times n$. The $i$th diagonal element of $\Sigma_{\bigcdot}$ will we written as $\sigma^{(i)}_{\bigcdot}$.

Throughout the paper, we will use $I_a$ to denote the identity matrix of size $a \times a$. We will also use $0$ to represent a matrix of zeros. For notational simplicity we omit the size of each zero matrix but note that it can be deduced from the notation. We will use $\otimes$ to represent the Kronecker product\footnote{
	The Kronecker product of $A\in\mathbb{R}^{m\times n}$
	and $B\in\mathbb{R}^{p\times r}$ is a $mp+nr$ block matrix where
	\begin{equation*}
	A\otimes B=
	\begin{bmatrix}
	A_{11}B & ... & A_{1n}B \\
	\vdots & \ddots & \vdots \\
	A_{m1}B & ... & A_{mn}B
	\end{bmatrix}.
	\end{equation*}
} and $\oplus$ to represent the Kronecker sum\footnote{The Kronecker sum is given by
	\begin{equation*}
	  A \oplus B = A \otimes I_b + I_a \otimes B
	\end{equation*}
	where $A$ is an $a\times a$ matrix and $B$ is a $b\times b$ matrix.
}.

\subsection{Discrete reaction-diffusion systems}\label{subsec:notation-two-classes}

We consider the discrete reaction-diffusion system on a one dimensional domain $[0,n]$ that is partitioned into $n$ equal-sized spatial compartments. Let $m$ denote the number of species (e.g., proteins or metabolites) in the system and $p$ denote the number of reactions. We will allow for three different boundary conditions:  homogeneous Neumann (no flux at both ends), Mixed (homogeneous Neumann at $x=0$ and open at $x=n$), and Open (flux allowed at both ends). Note that both the reactive and diffusive fluxes can be either positive or negative. We define a positive diffusive flux as moving in the positive $x$ direction. We will assume that all the species in the system diffuse at the same rate.

Within the domain there is a single barrier across which only a subset of species can diffuse. The barrier divides the system into two subregions where different reactions occur. Let $S_{r_1}\in \mathbb{R}^{m\times p_1}$ and $S_{r_2}\in \mathbb{R}^{m\times p_2}$ represent the stoichiometry matrices for the two subregions (i.e., $p_1$ reactions occur in the first subregion and $p_2$ reactions occur in the second). Note that the same reaction can occur in both regions.

We will let $n_1 \in \{1,...,n-1\}$ denote the number of compartments in the first subregion and $n_2 = n-n_1$ denote the number of compartments in the second subregion. Within this system, there are three diffusive processes: diffusion across the entire domain, within the first subregion, and within the second subregion. We define diffusion-only stoichiometry matrices for these three processes using $S_d\in\mathbb{R}^{n\times n+1}$, $S_{d_1}\in\mathbb{R}^{n_1\times n_1+1}$ and $S_{d_2}\in\mathbb{R}^{n_2\times n_2+1}$, respectively. For diffusion across the entire domain, we have that
\begin{equation}
   S_{d}:=
      \begin{bmatrix}
         b_1 & -1 & 0\\
         & 1 & -1\\
         &  & \ddots & \ddots\\
         &  &  & 1 & -1\\
         &  &  & 0 & 1 & -b_2
      \end{bmatrix}
      \label{eq:Sd}
\end{equation}
where the values in the first and last column depend on the boundary conditions. Specifically, $b_1=0$, $b_2=0$ implies zero flux boundary conditions, $b_1=0$, $b_2=1$ implies Mixed boundary conditions and $b_1=1$, $b_2=1$ implies Open boundary conditions. The diffusion only-stoichiometry matrices $S_{d_1}$ and $S_{d_2}$ are defined similarly. However,  for $S_{d_1}$ the value of $b_2$ is replaced by zero and for $S_{d_2}$ the value of $b_1$ is replaced by zero. The $n$ rows of $S_{d}$ corresponds to the species in each of the $n$ compartments, and the $n+1$ columns corresponds to the flux across each of the $n-1$ interior edges as well as the two boundaries at either end of the domain. Using the following matrix,
\begin{equation}\label{eq:H}
   \left(H\right)_{i,j} :=
      \begin{cases}
         -1 & i=n_1, \quad j=n_1+1\\
         1 & i=n_1+1, \quad j=n_1+1\\
         0 & \text{otherwise.}
      \end{cases},
\end{equation}
we can relate $S_d$ with $S_{d_1}$ and $S_{d_2}$ as follows
\begin{equation}
   \begin{bmatrix}
      S_{d_1}^{(1:n_1)} & 0 & 0\\
      0 & 0 & S_{d_2}^{(2:n_2+1)}
   \end{bmatrix} = S_d-H.
\end{equation}

Next, we provide definitions used to identify the species that can and cannot diffuse between the two subregions. When defining parameters (e.g., sets, matrices), a subscripted $+$ or $-$ will imply a relationship with the set of species that can ($+$) or cannot ($-$) diffuse across the barrier. The set $\mathcal{M}_+$ will contain indices for species that can diffuse across the barrier, whereas the set $\mathcal{M}_-$ will contain indices for species that cannot diffuse across the barrier. Additionally, let $m_+=|\mathcal{M}_+|$ and $m_-=|\mathcal{M}_-|$ where $m_++m_-=m$. Using these sets we define the diagonal matrices $D_+,D_- \in\mathbb{R}^{m\times m}$ where
\begin{align*}
   (D_+)_{i,j} :=
   \begin{cases}
   1 & i \in \mathcal{M}_+ \\
   0 & \text{otherwise}
   \end{cases}, \quad \quad
   (D_-)_{i,j} :=
   \begin{cases}
   1 & i \in \mathcal{M}_- \\
   0 & \text{otherwise}
   \end{cases}
\end{align*}
Note that $D_- + D_+ = I_m$.

We can now write the equation for the spatially-discrete RD stoichiometry matrix, given by (\ref{eq:S}). For convenience we rewrite this equation below
\begin{equation}
   S :=
      \begin{bmatrix}
         \gamma
         \begin{bmatrix}
            I_{n_1}\otimes S_{r_1} & 0 \\
            0 & I_{n_2} \otimes S_{r_2}
         \end{bmatrix}& S_{d}\otimes D_+ + (S_d - H)\otimes D_{-}
      \end{bmatrix}.
\end{equation}
The parameter $\gamma \ge 0$ represents the relative rate of the reactions compared with the rate of diffusion (i.e., if $\gamma \gg 1$ the reactions are much faster than diffusion and, if $\gamma\ll 1$, the reactions are much slower than diffusion). The first $n_1p_1+n_2p_2$ columns of $S$ correspond to the reactions occurring in each compartment. The final $(n+1)m$ columns correspond to the diffusion of species into or out of the domain as well as between adjacent compartments.

\subsection{Additional spatially-dependent parameters}\label{subsec:notation-sets}

Here we define the constants $C_1$, $C_2$, the set $\mathcal{B}$, and sets denoted by variations of $\mathcal{J}$. These parameters are only dependent on the spatial properties of the system (e.g., compartment number and boundary conditions), and are therefore unaffected if reactive properties (e.g., reaction number and stoichiometry) change.

The constants $C_1$ and $C_2$ depend on the boundary conditions and compartment numbers. We have that
\begin{equation}\label{eq:C12}
   C_1 := \left(\frac{2n_1 + b_1}{2n +b_1 + b_2}\right)^{1/2},
\quad \quad
   C_2 := \left(\frac{2n_2 + b_2}{2n+b_1+b_2}\right)^{1/2}.
\end{equation}
We will show in Lemma~\ref{lemma:eig-repeats} that these constants relate the singular vectors for $S_d$, $S_{d_1}$ and $S_{d_2}$ to one another.

Next, the set $\mathcal{B}$ is defined to contain indices that correspond to the columns of $S_d$ that are zero as well as the index of the column of $S_d$ that corresponds to the flux between the two subregions. Specifically,
\begin{equation}
   \mathcal{B} := \begin{cases}
         \{1,n_1+1,n+1\} & \text{Zero Flux} \\
         \{1,n_1+1\} & \text{Mixed} \\
         \{n_1+1\} & \text{Open}
      \end{cases}
\end{equation}
This set will be used to help define the nullspace of the RD stoichiometry matrix.

Finally, we define the following index sets of singular values for the diffusion-only stoichiometry matrices
\begin{equation}\label{eq:J}
   \begin{aligned}
   	\mathcal{J} &:= \{j \mid \sigma_d^{(j)} \in \text{diag}(\Sigma_d) \text{ and }
   	\sigma_d^{(j)} \notin \text{diag}(\Sigma_{d_1})\} &
   	\mathcal{J}^C &:= \{1,...,n\} \setminus \mathcal{J} \\
   	\mathcal{J}_1 &:= \{j \mid \sigma_{d_1}^{(j)} \in \text{diag}(\Sigma_{d_1}) \text{ and }
   	\sigma_{d_1}^{(j)} \notin \text{diag}(\Sigma_d)\} &
   	\mathcal{J}^C_1 &:= \{1,...,n_1\} \setminus \mathcal{J}_1 \\
   	\mathcal{J}_2 &:= \{j \mid \sigma_{d_2}^{(j)} \in \text{diag}(\Sigma_{d_2}) \text{ and }
   	\sigma_{d_2}^{(j)} \notin \text{diag}(\Sigma_d)\} &
   	\mathcal{J}^C_2 &:= \{1,...,n_2\} \setminus \mathcal{J}_2
   \end{aligned}
\end{equation}
and the analogous index sets for only nonzero singular values
\begin{equation}
   \begin{aligned}
      \hat{\mathcal{J}} &:= \{j \in \mathcal{J} \mid \sigma_d^{(j)} \ne 0\}
         &
      \hat{\mathcal{J}}^C &:= \{1,...,q_d\} \setminus \hat{\mathcal{J}} \\
      \hat{\mathcal{J}}_1 &:= \{j \in \mathcal{J}_1 \mid \sigma_{d_1}^{(j)} \ne 0\} &
      \hat{\mathcal{J}}_1^C &:= \{1,...,q_{d_1}\} \setminus \hat{\mathcal{J}}_1  \\
      \hat{\mathcal{J}}_2 &:= \{j \in \mathcal{J}_2 \mid \sigma_{d_2}^{(j)} \ne 0\} &
      \hat{\mathcal{J}}_2^C &:= \{1,...,q_{d_2}\} \setminus \hat{\mathcal{J}}_2
   \end{aligned}
\end{equation}
We will use these sets to define how singular values repeat in the system when $\gamma=0$. Understanding this property is a key step in proving Theorem~\ref{thm:svd-S}.

\subsection{Additional reaction-dependent stoichiometry-like matrices}\label{sub-sec:notation-more-matrices}
We refer to modified versions of the reaction-only stoichiometry matrices as \textit{stoichiometry-like matrices}. In this section we will define the $4 + |\hat{\mathcal{J}}|$ stoichiometry-like matrices that are necessary for writing the SVD. These matrices are given as $S_{r_1,-}$, $S_{r_2,-}$, $S_{\bar{r}}$ and $S_{\bar{r}_+,j}$ for $j \in \hat{\mathcal{J}}$.

The matrices $S_{r_1,-}$ and $S_{r_2,-}$ will represent subsetted versions of $S_{r_1}$ and $S_{r_2}$, respectively, that only contain rows for species that cannot diffuse across the boundary. Specifically,
\begin{equation}\label{eq:Sr-parts}
   \begin{aligned}
      S_{r_{1},+} &:= \left(S_{r_1}\right)_{\mathcal{M_+}} &
      S_{r_1,-} &:= \left(S_{r_1}\right)_{\mathcal{M_-}} \\
      S_{r_2,+} &:= \left(S_{r_2}\right)_{\mathcal{M_+}} &
      S_{r_2,-} &:= \left(S_{r_2}\right)_{\mathcal{M_-}}
   \end{aligned}
\end{equation}
where $(A)_\mathcal{B}$ represents the rows of $A$ that are in the index-set $\mathcal{B}$.

Next, we define a stoichiometry-like matrix that represents a merger of the two reaction-only stoichiometry matrices:
\begin{equation}\label{eq:S_rbar}
	S_{\bar{r}} :=
   \begin{bmatrix}
		C_1^2 I_m^{\mathcal{M}_+} S_{r_1,+} + I_m^{\mathcal{M}_-} S_{r_1,-} & C_1 C_2 I_m^{\mathcal{M}_+} S_{r_2,+} \\
		C_1 C_2 I_m^{\mathcal{M}_+} S_{r_1,+} & C_2^2 I_m^{\mathcal{M}_+} S_{r_2,+} + I_m^{\mathcal{M}_-} S_{r_2,-}
	\end{bmatrix}.
\end{equation}
To prove Theorem~\ref{thm:svd-S}, we will need to consider the eigendecomposition of
\begin{equation}
   B := S_{\bar{r}} S_{\bar{r}}^T.
\end{equation}
It can be shown that
\begin{equation}\label{eq:B}
	B =
   	\begin{bmatrix}
	   	B_2 + C_1^2 ( B_1 + B_4 + B_4^T) & C_1 C_2 ( B_1 + B_5 + B_4^T) \\
   		C_1 C_2 ( B_1 + B_4 + B_5^T) & B_3 + C_2^2 (B_1 + B_5 + B_5^T )
	\end{bmatrix}
\end{equation}
and
\begin{equation}\label{eq:Bi}
   \begin{aligned}
      B_1 &:= D_+ \left(C_1^2 S_{r_1} S_{r_1}^T + C_2^2 S_{r_2} S_{r_2}^T\right) D_+ \\
      B_2 &:= D_- \left(S_{r_1} S_{r_1}^T\right) D_- \\
      B_3 &:= D_- \left(S_{r_2} S_{r_2}^T\right) D_- \\
      B_4 &:= D_+ \left( S_{r_1} S_{r_1}^T \right) D_- \\
      B_5 &:= D_+ \left( S_{r_2} S_{r_2}^T \right) D_-.
   \end{aligned}
\end{equation}
To obtain this equation, we use that $I_m^{\mathcal{M}_+} S_{r_1,+} = D_+ S_{r_1}$ and similar identities.

Finally, for $j \in \hat{\mathcal{J}}$, define
\begin{equation}\label{eq:S_rbar_+_j}
	S_{\bar{r}_+,j} := \begin{bmatrix}
		|u_{d,s_1}^{(j)}| S_{r_1,+} & |u_{d,s_2}^{(j)}| S_{r_2,+}
	\end{bmatrix}.
\end{equation}
Note that $S_{\bar{r}_+,j}$ for $j \in \hat{\mathcal{J}}$ are the only stoichiometry-like matrices that depend on the spatial properties of the system.

\subsection{Additional SVD notation for stoichiometry-like and the diffusion-only stoichiometry matrices}\label{subsec:notation-SVD}

Generally, (\ref{eq:svd-Sstar}) will be used to write the SVDs of the stoichiometry and stoichiometry-like matrices. However, there are a few additional notational notes and one exception that will be discussed in this section.

First, the exception to this notational format will be for the left singular vectors of $S_{\bar{r}}$. Specifically, when considering the left nullspace of $S_{\bar{r}}$, we will exclude the space spanned by the following set of vectors
\begin{equation}\label{eq:breveU_ex}
   \breve{U}_{\bar{r},ex} :=
   \begin{bmatrix}
      C_2 I_m^{\mathcal{M}_+} \\
      -C_1 I_m^{\mathcal{M}_+}. \\
   \end{bmatrix}
\end{equation}
We define $\breve{U}_{\bar{r}} := \text{span}(\text{null}(S^T,\breve{U}_{\bar{r},ex}))$ and $U_{\bar{r}} := \begin{bmatrix}\hat{U}_{\bar{r}} & \breve{U}_{\bar{r}}\end{bmatrix}$. The reason for this will become clear in the proof to Theorem~\ref{thm:svd-S}.

In some instances, we divide a given singular vector into two components. We will use a subscripted $s_1$ or $s_2$ to refer to portions of the singular vectors that correspond to processes  that occur in the first or second subregion, respectively. Additionally, we wil use the subscript $m_1$ and $m_2$ to represent singular vectors that are divided into two subvectors of size $m$. More specifically, for the singular vectors of $S_d$, we have that
\begin{equation}
   u_d =
   \begin{bmatrix}
      u_{d,s_1} \\ u_{d,s_2}
   \end{bmatrix},
   \quad \quad
   v_d =
   \begin{bmatrix}
      v_{d,s_1} \\ v_{d,s_2}
   \end{bmatrix}
\end{equation}
where $u_{d,s_1} \in \mathbb{R}^{n_1}$, $u_{d,s_2} \in \mathbb{R}^{n_2}$, $v_{d,s_1} \in \mathbb{R}^{n_1}$, and $v_{d,s_2} \in \mathbb{R}^{n_2+1}$. For the singular vectors of $S_{\bar{r}}$ and the right singular vectors of $S_{\bar{r}_+,j}$, we define
\begin{equation}
   u_{\bar{r}} =
      \begin{bmatrix}
         u_{\bar{r},m_1} \\
         u_{\bar{r},m_2}
      \end{bmatrix}, \quad
   v_{\bar{r}} =
      \begin{bmatrix}
         v_{\bar{r},s_1} \\
         v_{\bar{r},s_2}
      \end{bmatrix}, \quad
   v_{\bar{r}_+,j} =
      \begin{bmatrix}
         v_{\bar{r}_+,j,s_1} \\
         v_{\bar{r}_+,j,s_2}
      \end{bmatrix}.
\end{equation}
where $u_{\bar{r},m_1},u_{\bar{r},m_2} \in \mathbb{R}^m$, $v_{\bar{r},s_1} \in \mathbb{R}^{p_1}$, $v_{\bar{r},s_2} \in \mathbb{R}^{p_2}$, $v_{\bar{r}_+,j,s_1}  \in \mathbb{R}^{p_1}$,  and $v_{\bar{r}_+,j,s_2} \in \mathbb{R}^{p_2}$. We will use the same notation to divide an entire set of right or left singular vectors into components. As an example, we have that
\begin{equation}
   U_d =
   \begin{bmatrix}
      U_{d,s_1} \\ U_{d,s_2}
   \end{bmatrix}, \quad
   \hat{U}_d =
   \begin{bmatrix}
      \hat{U}_{d,s_1} \\ \hat{U}_{d,s_2}
   \end{bmatrix}, \quad
   \breve{U}_d =
   \begin{bmatrix}
      \breve{U}_{d,s_1} \\ \breve{U}_{d,s_2}
   \end{bmatrix}.
\end{equation}

When considering the SVD of the diffusion-only stoichiometry matrices $S_d$, $S_{d_1}$ and $S_{d_2}$, the singular vectors and values can be written explicitly and depend on the specific boundary conditions (see Supplemental Material~\ref{sec-app:SVD-Sd}). The rank of $S_{d}$, given by $q_d$, also depends on the the boundary conditions where
\begin{equation*}
	q_{d} = n-1+b_2
\end{equation*}
and $b_2$ is as given in (\ref{eq:Sd}). This implies that the left nullspace, spanned by $\breve{U}_d$, is empty for both Mixed and Open boundary conditions.

\section{Preliminary Lemmas}\label{sec:useful-results}

In this section, we will provide preliminary lemmas that will be used to prove the main result.

First we consider the SVD of the diffusion-only stoichiometry matrices. In the following lemma we prove that, if a given singular value repeats across $\hat{\Sigma}_{d}$, $\hat{\Sigma}_{d_1}$, and $\hat{\Sigma}_{d_2}$, then it must be in all of these matrices. That is, a singular value will occur in either one or all three matrices.

\begin{lemma}\label{lemma:eig-repeats}
	Consider a system with Zero Flux, Mixed, or Open boundary conditions and the singular values defined in the matrices $\hat{\Sigma}_d$, $\hat{\Sigma}_{d_1}$, and $\hat{\Sigma}_{d_2}$. If a singular value is in two of these matrices then it is in all three.

   For singular values that are in all three matrices, the corresponding singular vectors are related as follows:
	\begin{align}
	u_d^{(j)} &=
	\begin{bmatrix}
	C_1 u_{d_1}^{(j_1)} \\
	(-1)^{n_1-j_1} C_2 u_{d_2}^{(j_2)}
	\end{bmatrix} \label{eq:Ud_divided}\\
	v_d^{(j)} &=
	\begin{bmatrix}
	C_1 \left(v_{d_1}^{(j_1)}\right)_{1:n_1} \\
	(-1)^{n_1-j_1} C_2 v_{d_2}^{(j_2)}
	\end{bmatrix} \label{eq:Vd_divided}
	\end{align}
   where $j$, $j_1$, and $j_2$ are such that $\sigma_d^{(j)} =\sigma_{d_1}^{(j_1)} = \sigma_{d_2}^{(j_2)}$ and $\left(v_{d_1}^{(j_1)}\right)_{1:n_1}$ represents the first $n_1$ entrees of $v_{d_1}^{(j_1)}$. Additionally, the indices $j$, $j_1$, and $j_2$ satisfy $j_1 = C_1^2 j$, $j_2 = C_2^2 j$, and $j=j_1 + j_2$.
\end{lemma}
\noindent The proof of this claim is given in Appendix~\ref{sec-app:proofs}.

We next derive formulas for the dimensions of the four fundamental subspaces of $S$, as given by (\ref{eq:S}). This allows us to verify that the SVD has the correct number of singular vectors in each space.

\begin{lemma}\label{lemma:rank}
	The rank of $S$, as given by (\ref{eq:S}), is
	\begin{equation}
   	q =
         \begin{cases}
      	  (n-2)m + m_+ + q_{\bar{r}} & \text{Zero Flux} \\
      	  (n-1)m + m_+ + q_{r_1,-} & \text{Mixed} \\
      	  nm & \text{Open}.
      	\end{cases}
	\end{equation}
	The dimension of the nullspace is
	\begin{equation}
   	\breve{q} =
         \begin{cases}
         	(n_1-1)p_1 + (n_2-1)p_2 + 3m - m_+ + \breve{q}_{\bar{r}} & \text{Zero Flux} \\
         	(n_1-1)p_1 + n_2 p _2 + 2m - m_+ + \breve{q}_{r_1,-} & \text{Mixed} \\
         	n_1 p_1 + n_2 p_2 + m & \text{Open}.
      	\end{cases}
	\end{equation}
	and the dimension of the left nullspace is
	\begin{equation}\label{eq:breve-q-Lb}
   	\breve{q}_{\ell} =
         \begin{cases}
         	m + m_- - q_{\bar{r}} & \text{Zero Flux} \\
         	m_- - q_{r_1,-} & \text{Mixed} \\
         	0 & \text{Open}.
      	\end{cases}
	\end{equation}
\end{lemma}
\noindent Here we are using the rank and nullspace size of $S_{r_1,-}$ and $S_{\bar{r}}$ defined in (\ref{eq:Sr-parts}) and (\ref{eq:S_rbar}). We omit the proof of this claim but note that it involves a sequence of row and column operations on $S$.

\section{Singular value decomposition derivation}\label{sec:SVD}

In this section, we present the proofs for Theorem \ref{thm:svd-S} and \ref{thm:svd-S-II}. Recall in Theorem \ref{thm:svd-S} we provide the approximate SVD for a system with a barrier, whereas in Theorem \ref{thm:svd-S-II} we consider a system without a barrier and derive an exact SVD for all relative diffusion/reaction time scales. We will also provide an alternative basis for the nullspace of the system with a barrier (Proposition~\ref{prop:breveV}).

To prove Theorem \ref{thm:svd-S}, we will apply concepts from linear perturbation theory and derive the SVD in the limit as diffusion becomes much faster than reactions. Specifically, we first consider the system at $\gamma=0$, and derive a set of left singular vectors and singular values (i.e., the eigenvectors and eigenvalues of $S S^T$ when $\gamma=0$). Because this system necessarily has repeating eigenvalues, the associated eigenvectors are not unique and are not necessarily continuous with respect to $\gamma$. However, we can apply results from Lemma \ref{lemma:eig-repeats} to find the unique orthonormal eigenprojection associated with each eigenvalue. Using these eigenprojections and perturbation theory results, we find the basis of eigenvectors that the system converges to continuously as $\gamma \rightarrow 0$. For a review of the necessary concepts from perturbation theory that are used in the proof see Appendix \ref{sec-app:perturbation-theory}.

To prove Theorem \ref{thm:svd-S-II}, we show directly that the given equations are equivalent to the SVD. We also show that the SVD given by Theorem \ref{thm:svd-S-II} is a simplified version of the SVD given by Theorem \ref{thm:svd-S} (see Corollary~\ref{cor:ItoII}).

\subsection{The perturbed and unperturbed systems}

The left singular vectors of $S$ are given by the solutions to the following
eigenvalue problem
\begin{equation*}\label{eq:SbSbT-eig}
	\begin{aligned}
		S &S^T u_i\\
		&=\left(\gamma^{2}
		\begin{bmatrix}
		I_{n_1} \otimes S_{r_1} S_{r_1}^T & 0 \\
		0 & I_{n_2} \otimes S_{r_2} S_{r_2}^T
		\end{bmatrix}
		+ S_{d}S_{d}^{T}\otimes D_+ + \left(S_{d}S_{d}^{T}-HH^{T}\right)\otimes D_-\right)u_{i} \\
		&=\lambda_{i}u_{i}.
	\end{aligned}
\end{equation*}
We will consider solutions to this eigenvalue problem in the limit as diffusion becomes much faster than reactions (i.e. $\gamma \rightarrow 0$). To consider this in the context of perturbation theory, we rewrite the eigenvalue problem as follows
\begin{equation}\label{eq:perturbationProblem}
   T(\gamma)u_i(\gamma) = (T+\gamma^2 T^{(1)})u_i(\gamma) = \lambda_i(\gamma)u_i(\gamma)
\end{equation}
where now we are explicitly including the dependency of $u_{i}$ and $\lambda_{i}$ on $\gamma$. The unperturbed matrix is
\begin{equation}\label{eq:T}
   T := S_{d}S_{d}^{T}\otimes D_+ + \left(S_{d}S_{d}^{T}-HH^{T}\right)\otimes D_-
\end{equation}
and the perturbation matrix is
\begin{equation}\label{eq:T-reac}
   T^{(1)} :=
   \begin{bmatrix}
   I_{n_1} \otimes S_{r_1} S_{r_1}^T & 0 \\
   0 & I_{n_2} \otimes S_{r_2} S_{r_2}^T.
   \end{bmatrix}.
\end{equation}
Given appropriate choices for the eigenvectors $u_{i}(\gamma)$, the eigenvectors and eigenvalues will be continuous functions of $\gamma$ in the neighborhood of $\gamma=0$.

\subsection{The eigenvalues and eigenprojections of the unperturbed system}

In this section we provide an orthonormal eigendecomposition for the unperturbed matrix $T$ (Lemma~\ref{lemma:unpertbEigval}). We then use this eigendecomposition along with the results from Lemma~\ref{lemma:eig-repeats} to find the unique orthonormal eigenprojections associated with each eigenvalue (Lemma~\ref{lemma:eigprojects}).

\begin{lemma}\label{lemma:unpertbEigval}
	An orthonormal eigendecomposition of $T$ is given as
	\begin{equation}
		T = Q_T \Lambda_T Q_T^T =
			\begin{bmatrix}
				\hat{Q}_T & \breve{Q}_T
			\end{bmatrix}
			\begin{bmatrix}
				\hat{\Lambda}_T & 0 \\
				0 & 0
			\end{bmatrix}
			\begin{bmatrix}
				\hat{Q}_T & \breve{Q}_T
			\end{bmatrix}^T
	\end{equation} where $\hat{\Lambda}_T$ contains the nonzero eigenvalues of $T$ and
	\begin{align}
		\hat{Q}_T &= \begin{bmatrix}
			\hat{Q}_{T,1} & \hat{Q}_{T,2} & \hat{Q}_{T,3}
		\end{bmatrix} \\
		\breve{Q}_T &=
		\begin{bmatrix}
			\breve{Q}_{T,1} & \breve{Q}_{T,2} & \breve{Q}_{T,3}
		\end{bmatrix} \\
		\hat{\Lambda}_T &= \begin{bmatrix}
			\hat{\Lambda}_{T,1} \\
			& \hat{\Lambda}_{T,2} \\
			& & \hat{\Lambda}_{T,3}
		\end{bmatrix}
	\end{align}where
	\begin{align}
		\label{eq:W}
		\hat{Q}_{T,1} &= \hat{U}_d \otimes I_m^{\mathcal{M}_+}, &
		\hat{Q}_{T,2} &=
		\begin{bmatrix}
			\hat{U}_{d_1} \\
			0
		\end{bmatrix}
		\otimes I_m^{\mathcal{M}_-}, &
		\hat{Q}_{T,3} &=
		\begin{bmatrix}
			0 \\
			\hat{U}_{d_2}
		\end{bmatrix} \otimes I_m^{\mathcal{M}_-} \\
		\label{eq:W-null}
		\breve{Q}_{T,1} &= \breve{U}_{d} \otimes I_m^{\mathcal{M}_+}, &
		\breve{Q}_{T,2} &=
		\begin{bmatrix}
			\breve{U}_{d_1} \\
			0
		\end{bmatrix}
		\otimes I_m^{\mathcal{M}_-}, &
		\breve{Q}_{T,3} &=
		\begin{bmatrix}
			0 \\
			\breve{U}_{d_2}
		\end{bmatrix}
		\otimes I_m^{\mathcal{M}_-} \\
		\label{eq:Sigma-W}
		\hat{\Lambda}_{T,1} &= \hat{\Sigma}_d^2 \otimes I_{m_+}, &
		\hat{\Lambda}_{T,2} &= \hat{\Sigma}_{d_1}^2 \otimes I_{m_-},  &
		\hat{\Lambda}_{T,3} &= \hat{\Sigma}_{d_2}^2 \otimes I_{m_-}.
	\end{align}
\end{lemma}
\noindent The proof of this lemma is given in Supplemental Material~\ref{sec-app:proofs}. From (\ref{eq:Sigma-W}) it is immediately clear that $T$ has repeating eigenvalues. This implies that the eigenvectors in the matrices given by (\ref{eq:W}) and (\ref{eq:W-null}) are not unique, and therefore, likely not the eigenvectors the system converges to as $\gamma \rightarrow 0$.

Lemma~\ref{lemma:eig-repeats} along with (\ref{eq:Sigma-W}) imply that eigenvalues of $T$ either repeat $m_+$, $m_-$, or $m + m_-$ times. Using this result and the set definitions defined in (\ref{eq:J}), we next identify each unique eigenvalue and find the associated orthonormal eigenprojection.
\begin{lemma}\label{lemma:eigprojects}
   The unique eigenvalues of $T$ are contained in the following three sets
   \begin{equation}\label{eq:uniqueEigs}
   	\begin{aligned}
   		&\left(\sigma_d^{(j)}\right)^2 \text{ for }j=1,...,n, &
   		&\left(\sigma_{d_1}^{(j)}\right)^2 \text{ for }j \in \mathcal{J}_1, &
   		&\left(\sigma_{d_2}^{(j)}\right)^2 \text{ for }j \in \mathcal{J}_2.
   	\end{aligned}
   \end{equation}
   The corresponding unique orthonormal projections are, respectively,
   \begin{align}
   	P_{j} &= u_d^{(j)} (u_d^{(j)})^T \otimes D_+ & j \in \mathcal{J}  \label{eq:P1} \\
      P_{j} & = u_d^{(j)} (u_d^{(j)})^T \otimes D_+ + \begin{bmatrix}
      u_{d_1}^{(j_1)}(u_{d_1}^{(j_1)})^T & 0 \\
      0 & u_{d_2}^{(j_2)}(u_{d_2}^{(j_2)})^T
      \end{bmatrix}
      \otimes D_- & j \in \mathcal{J}^C  \label{eq:P2}\\
      P_{{n_1},j} &=
   	\begin{bmatrix}
   		u_{d_1}^{(j)}(u_{d_1}^{(j)})^T & 0 \\
   		0 & 0
   	\end{bmatrix}
   	\otimes D_- & j \in \mathcal{J}_1 \label{eq:P3}\\
   	P_{n_2,j} &=
   	\begin{bmatrix}
   		0 & 0 \\
   		0 & u_{d_2}^{(j)}(u_{d_2}^{(j)})^T
   	\end{bmatrix}
   	\otimes D_- & j \in \mathcal{J}_2  \label{eq:P4}
   \end{align}
   where $j_1$ and $j_2$ are such that $\sigma_d^{(j)} =\sigma_{d_1}^{(j_1)} = \sigma_{d_2}^{(j_2)}$.
\end{lemma}

\begin{proof}
	From Lemma~\ref{lemma:unpertbEigval}, we see that every eigenvalue of $T$ is contained in the sets defined by (\ref{eq:uniqueEigs}) and from Lemma~\ref{lemma:eig-repeats} it follows that a given eigenvalue is only contained in one of the sets. Therefore,  (\ref{eq:uniqueEigs}) contains the unique eigenvalues of $T$.

	To find the orthonormal eigenprojection associated with each eigenvalue we will use the eigenvectors as defined in Lemma~\ref{lemma:unpertbEigval}. Specifically, using the eigenvectors given by (\ref{eq:W}) and (\ref{eq:W-null}), we can use (\ref{eq:sumProj}) in Appendix \ref{sec-app:perturbation-theory} to obtain the unique orthonormal eigenprojection.

	For $j \in \mathcal{J}$ the $m_+$ eigenvectors associated with $(\sigma_d^{(j)})^2$ are the columns of $u_d^{(j)} \otimes I_m^{\mathcal{M}_+}$ (see Lemma~\ref{lemma:unpertbEigval}). Therefore, the associated eigenprojection is
	\begin{equation}\label{eq:first-eigprojection}
		\begin{aligned}
		P_j &= \sum_{i\in \mathcal{M}_+} \left(u_d^{(j)} \otimes I_m^{(i)}\right)\left(u_d^{(j)} \otimes I_m^{(i)}\right)^T \\
			&= \sum_{i\in \mathcal{M}_+} u_d^{(j)} (u_d^{(j)})^T \otimes I_m^{(i)} (I_m)_i \\
			&= u_d^{(j)} (u_d^{(j)})^T \otimes D_+.
		\end{aligned}
	\end{equation}
	For $j \in \mathcal{J}^C$, there are $m$+$m_-$ eigenvectors associated with $(\sigma_d^{(j)})^2$. These eigenvectors are given by the columns of the following three matrices
	\begin{equation*}
		u_d^{(j)} \otimes I_m^{\mathcal{M}_+}, \quad \quad
		\begin{bmatrix}
			u^{(j_1)}_{d_1} \\
			0
		\end{bmatrix}
		\otimes I_m^{\mathcal{M}_-}, \quad \quad
		\begin{bmatrix}
			0 \\
			u^{(j_2)}_{d_2}
		\end{bmatrix}
		\otimes I_m^{\mathcal{M}_-},
	\end{equation*}
	where $j_1$ and $j_2$ are as given by Lemma~\ref{lemma:eig-repeats}. Using the same logic as shown in (\ref{eq:first-eigprojection}), we have that the eigenprojection can be written as
	\begin{equation}
		P_j =
		u_d^{(j)} (u_d^{(j)})^T \otimes D_+ + \begin{bmatrix}
		u_{d_1}^{(j_1)}(u_{d_1}^{(j_1)})^T & 0 \\
		0 & 0
		\end{bmatrix}
		\otimes D_- +
		\begin{bmatrix}
		0 & 0 \\
		0 & u_{d_2}^{(j_2)}(u_{d_2}^{(j_2)})^T
		\end{bmatrix}
		\otimes D_-.
	\end{equation}
	Applying analogous logic for $(\sigma_{d_1}^{(j)})^2$ for $j\in\mathcal{J}_1$ and $(\sigma_{d_2}^{(j)})^2$ for $j\in\mathcal{J}_2$ leads to the eigenprojections $P_{n_1,j}$ and $P_{n_2,j}$, respectively, as written in the claim.
\end{proof}

\subsection{The approximate left singular vectors of $S$}

We next use the eigenprojections given in Lemma~\ref{lemma:eigprojects} to derive the left singular vectors (i.e., eigenvectors of $T$) that the system converges to continuously as $\gamma \rightarrow 0$. This provides an approximate orthonormal basis for the column space and left null space of $S$ as $\gamma \rightarrow 0$.

\begin{proposition}\label{prop:left-singular-vectors}
   A complete set of left singular vectors of $S$ and corresponding singular values is given by the columns/diagonal elements of the following matrices:
   \begin{equation*}
      \begin{aligned}
         U_{1,j} &= u_d^{(j)} \otimes I_m^{\mathcal{M}_+} U_{\bar{r}_+,j} &
         \Sigma_{1,j}^2 &= \left(\sigma_d^{(j)}\right)^2 I_{m_+} +		\gamma^2 \left(\hat{\Sigma}_{\bar{r}_+,j}^2\right)_{m_+} \quad j \in \mathcal{J}\\
         U_2 &=
      		\begin{bmatrix}
      			U_{d_1}^{\mathcal{J}_1} \\
      			0
      		\end{bmatrix} \otimes I_m^\mathcal{M_-} U_{r_1,-} &
         \Sigma_2^2 &= (\Sigma_{d_1}^2)^{\mathcal{J}_1} \oplus \gamma^2 \left(\hat{\Sigma}_{r_1,-}^2\right)_{m_-} \\
         U_3 &=
      		\begin{bmatrix}
      			0 \\
      			U_{d_2}^{\hat{\mathcal{J}}_2}
      		\end{bmatrix}
      		\otimes I_m^{\mathcal{M}_-}U_{r_2,-} &
         \Sigma_3^2 &= (\Sigma_{d_2}^2)^{\hat{\mathcal{J}}_2} \oplus \gamma^2 \left(\hat{\Sigma}_{r_2,-}^2\right)_{m_-} \\
         U_4 &=
      		\begin{bmatrix}
      			\frac{1}{C_1} U_{d,s_1}^{\mathcal{J}^C} \otimes U_{\bar{r},m_1} \\
      			\frac{1}{C_2} U_{d,s_2}^{\mathcal{J}^C} \otimes U_{\bar{r},m_2}
      		\end{bmatrix} &
         \Sigma_4^2 &=
            (\Sigma_d^2)^{\mathcal{J}^C} \oplus
               \gamma^2\left(\hat{\Sigma}_{\bar{r}}^2\right)_{q_{\bar{r}}+\breve{q}_{\bar{r}}}.
		\end{aligned}
	\end{equation*}
\end{proposition}

Unlike Theorem \ref{thm:svd-S}, we are not identifying which singular vectors correspond to nonzero verse zero singular values. This is because the definitions provided in Proposition~\ref{prop:left-singular-vectors} allow for a direct comparison with the eigenprojections given by Lemma~\ref{lemma:eigprojects}. However, Proposition~\ref{prop:left-singular-vectors} immediately gives the left singular vectors in Theorem \ref{thm:svd-S}. To see this note that
\begin{align*}
   \hat{U}_{1,j} &= U_{1,j} \\
   \begin{bmatrix}\hat{U}_2 & \hat{U}_3 & \breve{U}_1 \end{bmatrix} &=  U_2 P_1 \\
   \hat{U}_4 &= U_3 \\
   \begin{bmatrix}\hat{U}_5 & \hat{U}_6 & \breve{U}_2\end{bmatrix} &= U_4 P_2
\end{align*}
where the matrices on the left of the equality represent those defined in Theorem~\ref{thm:svd-S} and the matrices on the right represent those defined in Proposition \ref{prop:left-singular-vectors}. The permutation matrices $P_1$ and $P_2$ are required to ensure the columns are in the correct order for comparison. Note that the singular values are related analogously.

\begin{proof}[Proof of Proposition \ref{prop:left-singular-vectors}]
	We will prove this result by considering the eigenprojections of $T$ defined in Lemma~\ref{lemma:eigprojects}. For each eigenprojection we calculate $\tilde{T}^{(1)}$, given be (\ref{eq:Ttilde}), and its eigendecomposition. We will write
	\begin{equation}\label{eq:T1star}
	\tilde{T}^{(1)}_{\bigcdot} = P_{\bigcdot} T^{(1)} P_{\bigcdot}
	\end{equation}
	where $\bigcdot$ depends on the eigenvalue/eigenprojection we are considering and $T^{(1)}$ is given by (\ref{eq:T-reac}). In the limit as $\gamma \rightarrow 0$, the eigenvectors of (\ref{eq:T1star}) in the range of $P_{\bigcdot}$ are equivalent to the left singular vectors. Additionally, the eigenvalues of (\ref{eq:T1star}) are used to find linear approximations of the singular values as shown by (\ref{eq:lambdakx}).

	First for $j \in \mathcal{J}$, consider the eigenprojection given by (\ref{eq:P1}). We have that
	\begin{equation}\label{eq:Tj1}
	\tilde{T}_j^{(1)} = P_j T^{(1)} P_j.
	\end{equation}
	To find the eigendecomposition of (\ref{eq:Tj1}), we apply Property~\ref{app-property:Kron1} given in Supplemental Material \ref{sec-app:Kronecker-formulas} to show that
	\begin{align*}\label{eq:Qj-J}
	\tilde{T}_j^{(1)} &= u_d^{(j)} (u_d^{(j)})^T \otimes D_+ \left(|u_{d,\mathcal{N}_1}^{(j)}|^2 S_{r_1} S_{r_1}^T + |u_{d,\mathcal{N}_2}^{(j)}|^2 S_{r_2} S_{r_2}^T\right) D_+ \\
	&= u_d^{(j)} (u_d^{(j)})^T \otimes I_m^{\mathcal{M}_+} S_{\bar{r}_+,j} S_{\bar{r}_+,j}^T I_{m,\mathcal{M}_+}
	\end{align*}
	where recall $S_{\bar{r}_+,j}$ is given by (\ref{eq:S_rbar_+_j}). The eigenvectors of $\tilde{T}_j^{(1)}$ in the range of $P_j$ are the columns of
	\begin{equation}
		\hat{U}_{1,j} = u_d^{(j)} \otimes I_m^{\mathcal{M}_+} U_{\bar{r}_+,j}.
	\end{equation}
	and the corresponding eigenvalues are contained in $(\Sigma^2_{\bar{r}_+,j})_{m_+}$. Using (\ref{eq:lambdakx}) this leads to the following linear approximation of the singular values
	\begin{equation}
		\hat{\Sigma}_{1,j} = \sqrt{\left(\sigma_d^{(j)}\right)^2 I_{m_+} + \gamma^2 \left(\Sigma^2_{\bar{r}_+,j}\right)_{m_+}}.
	\end{equation}

	Similarly, consider the eigenprojections given by (\ref{eq:P3}) and (\ref{eq:P4}). Using \ref{eq:T1star} and Property~\ref{app-property:Kron2} given in Appendix \ref{sec-app:Kronecker-formulas}, we have that
	\begin{equation}
		\tilde{T}^{(1)}_{n_1,j} =
			\begin{bmatrix}
			u_{d_1}^{(j)}(u_{d_1}^{(j)})^T & 0 \\
			0 & 0
			\end{bmatrix} \otimes
			D_- \left(S_{r_1} S_{r_1}^T\right) D_-
	\end{equation}
	and
	\begin{equation}
		\tilde{T}^{(1)}_{n_2,j} =
			\begin{bmatrix}
				0 & 0 \\
				0 & u_{d_2}^{(j)}(u_{d_2}^{(j)})^T
			\end{bmatrix} \otimes
			D_- \left(S_{r_2} S_{r_2}^T\right) D_-.
	\end{equation}
	The eigenvectors of $\tilde{T}^{(1)}_{n_1,j}$ and $\tilde{T}^{(1)}_{n_2,j}$ that are in the range of $P_{n_1,j}$ and $P_{n_2,j}$ are
	\begin{equation}
	U_{n_1,j} =
		\begin{bmatrix}
		u_{d_1}^{(j)} \\
		0
		\end{bmatrix}
		\otimes I_m^{\mathcal{M}_-} U_{r_1,-} \quad \text{and} \quad
		U_{n_2,j} =
		\begin{bmatrix}
			0 \\
			u_{d_2}^{(j)}
		\end{bmatrix}
		\otimes I_m^{\mathcal{M}_-}U_{r_2,-}
	\end{equation}
	and the corresponding eigenvalues are $\Sigma^2_{r_1,-}$ and $\Sigma^2_{r_2,-}$, respectively. Recall the definition of $S_{r_i,-}$ is given by (\ref{eq:Sr-parts}). Again, using (\ref{eq:lambdakx}) this leads to the singular values given by $\Sigma_2$ and $\Sigma_3$.

	Finally, suppose $j\in \mathcal{J}^C$ and consider the eigenprojection $P_j$ given by (\ref{eq:P2}). Let $j_1$ and $j_2$ be as given by Lemma~\ref{lemma:eig-repeats}. For notational simplicity we will make the following substitutions
	\begin{align*}
		u &= u_d^{(j)} \\
		u_1 &= u_{d_1}^{(j_1)} \\
		u_2 &= u_{d_2}^{(j_2)} \\
		\alpha &= (-1)^{n_1-j_1}
	\end{align*}
	We will also use the matrices $B_i$ for $i=1,...,5$ given by (\ref{eq:Bi}). Using Property~\ref{app-property:Kron1} and \ref{app-property:Kron2} given in Appendix \ref{sec-app:Kronecker-formulas}, the eigenvector relationship given by (\ref{eq:Ud_divided}) in Lemma~\ref{lemma:eig-repeats}, and (\ref{eq:T1star}) we have that
	\begin{align*}
		\tilde{T}^{(1)}_j &= u u^T \otimes B_1
			+ \begin{bmatrix}
				u_1 u_1^T & 0 \\
				0 & 0
			\end{bmatrix} \otimes B_2
			+ \begin{bmatrix}
				0 & 0 \\
				0 & u_2 u_2^T
			\end{bmatrix} \otimes B_3 \\
			&+ \begin{bmatrix}
				C_1^2 u_1 u_1^T  &
				0 \\
				\alpha C_1 C_2 u_2 u_1^T &
				0
			\end{bmatrix} \otimes B_4
			+ \begin{bmatrix}
				0 &
				\alpha C_1 C_2 u_1 u_2^T  \\
				0 &
				C_2^2 u_2 u_2^T
			\end{bmatrix} \otimes B_5 \\
			&+ \begin{bmatrix}
				C_1^2 u_1 u_1^T &
				\alpha C_1 C_2 u_1 u_2^T \\
				0 &
				0
			\end{bmatrix} \otimes B_4^T
			+ \begin{bmatrix}
				0 &
				0 \\
				\alpha C_1 C_2 u_2 u_1^T &
				C_2^2 u_2 u_2^T
			\end{bmatrix} \otimes B_5^T.
	\end{align*}

	To obtain the eigendecomposition of $\tilde{T}_j^{(1)}$ for $j \in \mathcal{J}^C$, we will suppose that the eigenvectors take the form
	\begin{equation}
		v =
			\begin{bmatrix}
				u_1 \otimes v_1 \\
				\alpha u_2 \otimes v_2
			\end{bmatrix}.
	\end{equation}
	and derive the values of $v_1, v_2 \in \mathbb{R}^{m\times 1}$.

	We have that
		\begin{multline}
			\tilde{T}^{(1)}_j v =
				\begin{bmatrix}
					u_1 \otimes (C_1^2 B_1 v_1 + C_1 C_2 B_1 v_2)\\
					\alpha u_2 \otimes (C_1 C_2 B_1 v_1 + C_2^2 B_1 v_2)
				\end{bmatrix}
				+
				\begin{bmatrix}
					u_1 \otimes B_2 v_1 \\
					0
				\end{bmatrix}
				+
				\begin{bmatrix}
					0 \\
					\alpha u_{2} \otimes B_3 v_2
				\end{bmatrix} \\
				+
				\begin{bmatrix}
					u_{1} \otimes C_1^2 B_4 v_1 \\
					\alpha u_{2} \otimes C_1 C_2 B_4 v_1
				\end{bmatrix}
				+
				\begin{bmatrix}
					u_{1} \otimes C_1 C_2 B_5 v_2 \\
					\alpha u_{2} \otimes C_2^2 B_5 v_2
				\end{bmatrix} \\
				+
				\begin{bmatrix}
					u_{1} \otimes B_4^T (C_1^2 v_1 +  C_2^2 v_2) \\
					0
				\end{bmatrix}
				+
				\begin{bmatrix}
					0 \\
					\alpha u_{2} \otimes B_5^T (C_1^2 v_1 + C_2^2 v_2)
				\end{bmatrix}
				= \lambda
				\begin{bmatrix}
					u_{1} \otimes v_1 \\
					\alpha u_{2} \otimes v_2
				\end{bmatrix}.
		\end{multline}
   where we are using Property~\ref{app-property:Kron3} to calculate $(uu^T \otimes B_1)v$. Therefore, for $v$ to be an eigenvector of $\tilde{T}_j^{(1)}$ the following smaller eigenvalue problem must hold
	\begin{equation*}
		B
		\begin{bmatrix}
			v_1 \\ v_2
		\end{bmatrix} = \lambda
		\begin{bmatrix}
			v_1 \\ v_2
		\end{bmatrix}.
	\end{equation*}
   where $B$ is given by (\ref{eq:B}). This implies that $[v_1;v_2]$ is equal to a left singular vector of $S_{\bar{r}}$, given by (\ref{eq:S_rbar}). Note that since $[v_1;v_2]$ is a unit vector, $v$ is also a unit vector and, thus, properly normalized.

   Only some of the singular vectors of $S_{\bar{r}}$ result in eigenvectors $v$ that are in the range of $P_j$. Specifically, note the singular vectors contained in the columns of $\breve{U}_{\bar{r},ex}$, see (\ref{eq:breveU_ex}), result in eigenvectors that are not in the range of $P_j$. To see this note that, when $[v_1,v_2] = \breve{U}_{\bar{r},ex}$,
   \begin{equation*}
   \begin{aligned}
      P_j v &= \left(u u^T \otimes D_+ + \begin{bmatrix}
         u_1 u_1^T & 0 \\
         0 & u_2 u_2^T
         \end{bmatrix}
         \otimes D_- \right)
         \begin{bmatrix}
   			u_1 \otimes C_2 I_m^{\mathcal{M}^+} \\
   			\alpha u_2 \otimes -C_1 I_m^{\mathcal{M}^+}
   		\end{bmatrix} \\
      &= \begin{bmatrix}
         C_1^2 u_1 \otimes C_2 I_m^{\mathcal{M}^+} - C_1 C_2 u_1 \otimes C_1 I_m^{\mathcal{M}^+} \\
         \alpha C_1 C_2 u_2 \otimes C_2 I_m^{\mathcal{M}^+} - \alpha C_2^2 u_2 \otimes C_1 I_m^{\mathcal{M}^+}
      \end{bmatrix} \\
      &= 0.
   \end{aligned}
   \end{equation*}

   Coupled with the left singular vector relationship given by (\ref{eq:Ud_divided}), this completes our derivation of the singular vectors contained in $U_4$. Using (\ref{eq:lambdakx}), we obtain the singular values given by $\Sigma_4$.
\end{proof}

\subsection{Right singular vectors}

Next we will approximate the right singular vectors of the system in the limit as $\gamma \rightarrow 0$. To derive the right singular vectors that represent a basis for the row space, we use the following equation and the results from Proposition \ref{prop:left-singular-vectors}. For $i=1,...,5$,
\begin{equation}\label{eq:V-from-U}
	\hat{V}_i = S^T \hat{U}_i \hat{\Sigma}_i^{-1}
\end{equation}
Note that the equations for $\hat{U}_i$ and $\hat{\Sigma}_i$ are given by Theorem~\ref{thm:svd-S}, however their derivation is found in the proof to Proposition \ref{prop:left-singular-vectors}. Using this equation, we obtain the set of right singular vectors given by Theorem \ref{thm:svd-S}.

To complete the proof of Theorem \ref{thm:svd-S}, it remains to show that $\breve{V}$ defines a orthonormal basis for the nullspace that the system approaches as $\gamma \rightarrow 0$. The complete proof of this is given in Supplemental Material~\ref{sec-app:proofs}. Note that, an alternative asymptotic nullspace can be found. The nullspace given by Theorem \ref{thm:svd-S} has the property that it is orthogonal for small values of $\gamma$ and $S \breve{V} \rightarrow 0$ as $\gamma \rightarrow 0$. It is possible to instead find a basis such that $S \breve{V} = 0$ for small values of $\gamma$ and the basis approaches orthogonal in the limit as $\gamma \rightarrow 0$. The following lemma provides the equations for this alternative basis.

\begin{proposition}\label{prop:breveV}
	The column vectors in the follow matrices span the nullspace of $S$ as given by (\ref{eq:S}) and this basis is orthogonal in the limit as $\gamma \rightarrow 0$:
   \begin{equation}
      \breve{V} = \begin{bmatrix}
         \breve{V}_1 & \breve{V}_2 &	\breve{V}_3 & \breve{V}_4 & \breve{V}_5
      \end{bmatrix}
   \end{equation}
   where
   \begin{equation}
   	\centering
      \begin{aligned}
      \breve{V}_1 &=
         \left[\begin{array}{cc}
            U_{d_1}^{n\setminus (\mathcal{J}^C\cup\hat{\mathcal{J}}_1)} \otimes \breve{V}_{r_1,-} \\
            0 \\
            \hdashline \noalign{\vskip 3pt}
            -\gamma w_1 \otimes S_{r_1} \breve{V}_{r_1,-}
         \end{array}\right](I_{\breve{q}_{d_1}\breve{q}_{r_1,-}} + \gamma^2 W_1)^{-\frac{1}{2}}  \\
      \breve{V}_2 &=
         \left[\begin{array}{c}
            \hat{U}_{d_1} \hat{\Sigma}_{d_1} \otimes V_{r_1} \\
            0 \\
            \hdashline \noalign{\vskip 3pt}
            -\gamma \hat{V}_{d_1} \otimes U_{r_1} \Sigma_{r_1} \\
            0
         \end{array}\right] \left(\hat{\Sigma}_{d_1}^2 \oplus \gamma^2 \left(\hat{\Sigma}_{r_1}^2\right)_{p_1}\right)^{-\frac{1}{2}} \\
      \breve{V}_3 &=
         \left[\begin{array}{c} \vspace{5pt}
            \frac{1}{C_1} \breve{U}_{d,s_1} \otimes \breve{V}_{\bar{r},s_1} \\
            \frac{1}{C_2} \breve{U}_{d,s_2} \otimes \breve{V}_{\bar{r},s_2} \\
            \hdashline \noalign{\vskip 3pt}
            -\gamma w_2 \otimes S_{r_1}  \breve{V}_{\bar{r},s_1}
         \end{array}\right](I_{\breve{q}_d\breve{q}_{\bar{r}}} + \gamma^2 W_2)^{-\frac{1}{2}} \\
      \breve{V}_4 &=
         \left[\begin{array}{c}
            0 \\
            \hat{U}_{d_2} \hat{\Sigma}_{d_2} \otimes V_{r_2} \vspace{3pt}\\
            \hdashline \noalign{\vskip 2pt}
            0 \\
            -\gamma \hat{V}_{d_2} \otimes U_{r_2} \Sigma_{r_2}
         \end{array}\right] \left(\hat{\Sigma}_{d_2}^2 \oplus \gamma^2 \left(\hat{\Sigma}_{r_2}^2\right)_{p_2}\right)^{-\frac{1}{2}} \\
     \breve{V}_5 &=
         \left[\begin{array}{cc}
            0 & 0\\
            \hdashline \noalign{\vskip 3pt}
             \breve{V}_{d} \otimes I_m^{\mathcal{M}_+} & I_{n+1}^{\mathcal{B}} \otimes I_m^{\mathcal{M}_-}
         \end{array}\right].
   	\end{aligned}
   \end{equation}
   where
   \begin{align*}
      w_1 &= \hat{V}_d \hat{\Sigma}_d^{-1} \hat{U}_d^T
      \begin{bmatrix}
         U_{d_1}^{n\setminus (\mathcal{J}^C\cup\hat{\mathcal{J}}_1)} \\
         0
      \end{bmatrix} &
      \left(W_1\right)_{ii} &= w_2^T w_2 \left(
      S_{r_1} \breve{V}_{r_1,-}^{(i)}
      \right)^T S_{r_1} \breve{V}_{r_1,-}^{(i)}.
   \end{align*}
	and
	\begin{align*}
		w_2 &= \frac{n}{n_2 + \sqrt{n_1 n_2}}\hat{V}_d \hat{\Sigma}_d^{-1} \hat{U}_d^T
		\begin{bmatrix}
			\frac{1}{C_1}\breve{U}_{d,s_1} \\
			-\frac{1}{C_2}\breve{U}_{d,s_2}
		\end{bmatrix} &
		\left(W_2\right)_{ii} &= w_1^T w_1 \left(
		S_{r_1}  \breve{V}_{\bar{r},s_1}^{(i)}
		\right)^T S_{r_1} \breve{V}_{\bar{r},s_1}^{(i)}.
	\end{align*}
\end{proposition}

Notice that only $\breve{V}_1$ and $\breve{V}_3$ have changed when compared to Theorem~\ref{thm:svd-S}.

\subsection{SVD for systems with spatially homogeneous reactions and diffusion}\label{sec:SVD-1}

In the previous section we presented the approximate SVD for a system with a spatial barrier. Here, we will consider a specific scenario where there is no barrier and the reactions are the same across the domain. In terms of the previous notation, this is equivalent to setting $m_+ = m$, $m_-=0$ and $S_{r_1}=S_{r_2}$. Under these conditions, we will show that the SVD reduces to a simplified form (Corollary~\ref{cor:ItoII} and \ref{cor:ItoIIb}) and becomes exact for all values of $\gamma$, i.e., prove Theorem~\ref{thm:svd-S-II}. Below we set $S_r = S_{r_1}$ and refer to the singular value decomposition of $S_r$ using the notation given in (\ref{eq:svd-Sstar}).

First note, that under these conditions the SVDs of the stoichiometry-like matrices are simplified. We have that the SVD of $S_{\bar{r}_+,j}$ is
   \begin{equation}\label{eq:simp-SVD-1}
      U_{\bar{r}_+,j} = U_r, \quad \quad V_{\bar{r}_+,j} = \begin{bmatrix} |u_{d,s_1}^{(j)}| V_r \\ |u_{d,s_2}^{(j)}| V_r \end{bmatrix}, \quad \quad \hat{\Sigma}_{\bar{r}_+,j} = \hat{\Sigma}_r
   \end{equation}
   and the SVD of $S_{\bar{r}}$ is
   \begin{equation}\label{eq:simp-SVD-2}
      U_{\bar{r}} = \begin{bmatrix}
			C_1 U_r \\
			C_2 U_r
		\end{bmatrix}, \quad \quad V_{\bar{r}} = \begin{bmatrix}
			C_1 V_r \\
			C_2 V_r
		\end{bmatrix}, \quad \quad
      \hat{\Sigma}_{\bar{r}} = \hat{\Sigma}_r.
   \end{equation}
This result is shown by considering the equations for $S_{\bar{r}}$ and $S_{\bar{r}_+,j}$ as given by (\ref{eq:S_rbar}) and (\ref{eq:S_rbar_+_j}), respectively. Using (\ref{eq:simp-SVD-1}) and (\ref{eq:simp-SVD-2}), we next show that the SVD given by Theorem~\ref{thm:svd-S} reduces to a simplified form.

\begin{corollary}\label{cor:ItoII}
	The left singular vectors given by Theorem~\ref{thm:svd-S} reduce to the columns of following matrix
	\begin{equation}
		U = U_d \otimes U_r
	\end{equation}
	and the singular values reduce to the diagonal of
	\begin{equation}
		\Sigma = \Sigma_d \otimes \Sigma_r.
	\end{equation}
\end{corollary}

\begin{proof}
	To prove this corollary we will examine the SVD for a system with a barrier. Specifically, we consider the left singular vectors and values as written in Proposition~\ref{prop:left-singular-vectors}. We will show that these vectors and values reduce to the singular vectors and values given in the corollary statement

First, note that $U_2$ and $U_3$ are empty matrices since $\mathcal{M}_-$ is an empty set. For $U_{1,j}$, $j \in \mathcal{J}$, using that $I_m^{\mathcal{M}_+} = I_m$. and the SVD given by \ref{eq:simp-SVD-1} we have that,
	\begin{equation}
		\hat{U}_{1,j} = u_d^{(j)} \otimes U_r \quad \quad \Sigma_{1,j}^2 = \left( \sigma_d^{(j)} \right)^2 I_m + \gamma^2 \Sigma_r^2.
	\end{equation}
For $U_4$ using the SVD given by \ref{eq:simp-SVD-2} and that $q_{\bar{r}}=q_r$ and $\breve{q}_{\bar{r}}=m-q_r$, we have that
	\begin{equation}
		U_4 = U_d^{\mathcal{J}^C} \otimes U_r \quad \quad \Sigma_4 = (\Sigma_d^2)^{\mathcal{J}^C} \oplus \gamma^2 \left(\hat{\Sigma}_r^2\right)_m.
	\end{equation}
	Putting these results together we obtain the set of left singular vectors and singular values given by the corollary statement.
\end{proof}

\begin{corollary}\label{cor:ItoIIb}
   The right singular vectors of $S$ given by Theorem~\ref{thm:svd-S} reduce to the following for the simplified system.
   \begin{align}
      \hat{V} &=
   	\begin{bmatrix}
   		\gamma\left(\hat{U}_{d}\otimes \hat{V}_{r}\hat{\Sigma}_{r}\right)\tilde{\Sigma}^{-1} & 0 & \breve{U}_{d}\otimes \hat{V}_{r}\\
   		\left(\hat{V}_{d}\hat{\Sigma}_{d}\otimes \hat{U}_{r}\right)\tilde{\Sigma}^{-1} & \hat{V}_{d}\otimes \breve{U}_{r} & 0
   	\end{bmatrix} \\
   	\breve{V} &=
   	\begin{bmatrix}
   		\left(\hat{U}_{d}\hat{\Sigma}_{d}\otimes \hat{V}_{r}\right)\tilde{\Sigma}^{-1} & \hat{U}_{d}\otimes \breve{V}_{r} & 0\\
   		-\gamma\left(\hat{V}_{d}\otimes \hat{U}_{r}\hat{\Sigma}_{r}\right)\tilde{\Sigma}^{-1} & 0 & \breve{V}_{d}\otimes U_{r}
   	\end{bmatrix}
   \end{align}
\end{corollary}

\begin{proof}
   For the right singular vectors we will use the equations as given in Theorem~\ref{thm:svd-S}. We will show the proof for $\check{V}$ and note that the proof for $\breve{V}$ follows analogously. Note that $\hat{V}_2$, $\hat{V}_3$, and $\hat{V}_4$ are empty.

   For $\hat{V}_{1,j}$, using the SVD given by (\ref{eq:simp-SVD-1}), we have that
   \begin{align}
      \hat{V}_{1,j} &= \left[
         \begin{array}{c}
            \frac{\gamma}{|u_{d,s_1}^{(j)}|}  u_{d,s_1}^{(j)} \otimes  V_{\bar{r}_+,j,s_1} \Sigma_{\bar{r}_+,j}^T \\
            \frac{\gamma}{|u_{d,s_2}^{(j)}|} u_{d,s_2}^{(j)} \otimes  V_{\bar{r}_+,j,s_2} \Sigma_{+,j}^T
            \vspace{3pt}\\
            \hdashline \noalign{\vskip 3pt}
            v_d^{(j)} \sigma_d^{(j)} \otimes
            I_m^{\mathcal{M}_+} U_{\bar{r}_+,j}
         \end{array}\right] \hat{\Sigma}_{1,j}^{-1} \\
      &= \left[\begin{array}{c}
         \gamma u_{d}^{(j)} \otimes  V_{r} \Sigma_r^T
         \vspace{3pt}\\
         \hdashline \noalign{\vskip 3pt}
         v_d^{(j)} \sigma_d^{(j)} \otimes
         U_r
      \end{array}\right] \left(\left(\sigma_d^{(j)}\right)^2 I_{m} + \gamma^2 \left(\Sigma_r^2\right)_{m}\right)^{-1/2}.
   \end{align}

For $\hat{V}_5$ and $\hat{V}_6$, using the SVD given by (\ref{eq:simp-SVD-2}), we have that
\begin{align*}
   \hat{V}_5 &=
      \left[\begin{array}{cc} \vspace{5pt}
         \frac{\gamma}{C_1} U_{d,s_1}^{\hat{\mathcal{J}}^C} \otimes
         \begin{bmatrix}
            \hat{V}_{\bar{r},s_1} \hat{\Sigma}_{\bar{r}} & 0
         \end{bmatrix} \\
         \frac{\gamma}{C_2} U_{d,s_2}^{\hat{\mathcal{J}}^C} \otimes
         \begin{bmatrix}
            \hat{V}_{\bar{r},s_2} \hat{\Sigma}_{\bar{r}} & 0
         \end{bmatrix} \vspace{3pt} \\
         \hdashline \noalign{\vskip 3pt}
         \frac{1}{C_1}V_{d,s_1}^{\hat{\mathcal{J}}^C} \Sigma_d^{\hat{\mathcal{J}}^C} \otimes
         U_{\bar{r},m_1} \\
         \frac{1}{C_2}V_{d,s_2}^{\hat{\mathcal{J}}^C} \Sigma_d^{\hat{\mathcal{J}}^C} \otimes
         U_{\bar{r},m_2}
      \end{array}\right] \hat{\Sigma}_5^{-1} \\
   &=
      \left[\begin{array}{cc} \vspace{5pt}
         \gamma U_{d}^{\hat{\mathcal{J}}^C} \otimes
            V_r \Sigma_r \\
         \hdashline \noalign{\vskip 3pt}
         V_{d}^{\hat{\mathcal{J}}^C} \Sigma_d^{\hat{\mathcal{J}}^C} \otimes
         U_r
      \end{array}\right] \left((\hat{\Sigma}_d^2)^{\hat{\mathcal{J}}^C} \oplus
         \gamma^2\left(\hat{\Sigma}_r^2\right)_{m}\right)^{-1/2}. \\
   \hat{V}_6 &=
      \left[\begin{array}{cc} \vspace{5pt}
         \frac{\gamma}{C_1} \breve{U}_{d,s_1} \otimes \hat{V}_{\bar{r},s_1} \hat{\Sigma}_{\bar{r}} \\
         \frac{\gamma}{C_2} \breve{U}_{d,s_2} \otimes \hat{V}_{\bar{r},s_2} \hat{\Sigma}_{\bar{r}} \\
         \hdashline \noalign{\vskip 3pt}
         0
      \end{array}\right] \hat{\Sigma}_6^{-1} \\
   &=
      \left[\begin{array}{cc} \vspace{5pt}
         \gamma \breve{U}_{d} \otimes \hat{V}_{r} \hat{\Sigma}_r \\
         \hdashline \noalign{\vskip 3pt}
         0
      \end{array}\right] \left(
         I_{n-q_d} \otimes \gamma^2 \hat{\Sigma}_{r}^2\right)^{-1/2} \\
   &=
      \left[\begin{array}{cc} \vspace{5pt}
         \gamma \breve{U}_{d} \otimes \hat{V}_{r} \\
         \hdashline \noalign{\vskip 3pt}
         0
      \end{array}\right]. \\
\end{align*}
Putting this together and rearranging columns we obtain the equation for $\hat{V}$ given in the corollary statement.
\end{proof}

Finally, we will prove the main result that the SVD of the simplified $S$ is valid for all values of $\gamma$.

\begin{proof}[Proof of Theorem \ref{thm:svd-S-II}]

	To show that (\ref{eq:SVDFull}) is the SVD of $S$ as given by (\ref{eq:S_II}), it is suffices to show that, $U$ and $V$ are orthogonal matrices and $S =\hat{U}\hat{\Sigma}\hat{V}^T$.

	We first show that $U$ and $V$ are orthogonal matrices. Recall that $U \in \mathbb{R}^{nm \times nm}$ where $U = U_d \otimes U_r$. We have that
	\begin{equation}
      (U_d \otimes U_r)^T (U_d \otimes U_r) = U_d^T U_d \otimes U_r^T U_r = I_{nm}
	\end{equation}
	It follows that $U$ is orthogonal.

   For $V$, we have that
   \begin{align*}
      \hat{V}^T \hat{V}
      &=\left[\begin{smallmatrix}
   		\gamma\tilde{\Sigma}^{-1}\left(\hat{U}_{d}^T\otimes \hat{\Sigma}_{r} \hat{V}_{r}^T\right)
         & \tilde{\Sigma}^{-1}\left(\hat{\Sigma}_{d}\hat{V}_{d}^T\otimes \hat{U}_{r}^T\right)  \\
         0 & \hat{V}_{d}^T\otimes \breve{U}^T_{r} \\
         \breve{U}_{d}^T\otimes \hat{V}_{r}^T & 0
   	\end{smallmatrix}\right]
      \left[\begin{smallmatrix}
   		\gamma\left(\hat{U}_{d}\otimes \hat{V}_{r}\hat{\Sigma}_{r}\right)\tilde{\Sigma}^{-1} & 0 & \breve{U}_{d}\otimes \hat{V}_{r}\\
   		\left(\hat{V}_{d}\hat{\Sigma}_{d}\otimes \hat{U}_{r}\right)\tilde{\Sigma}^{-1} & \hat{V}_{d}\otimes \breve{U}_{r} & 0
   	\end{smallmatrix}\right]\\
      &=
         \begin{bmatrix}
            \gamma^2 \tilde{\Sigma}^{-1} \left(\hat{U}_d^T\hat{U}_{d}\otimes \hat{\Sigma}_r^2 + \hat{\Sigma}_{d}^2 \otimes \hat{U}_{r}^T \hat{U}_r\right)\tilde{\Sigma}^{-1} &
            0 & 0 \\
            0 & \hat{V}_{d}^T \hat{V}_d \otimes \breve{U}^T_{r} \breve{U}_r & 0 \\
            0 & 0 & \breve{U}_{d}^T \breve{U}_d \otimes \hat{V}_{r}^T \hat{V}_r
         \end{bmatrix} \\
      &= I
   \end{align*}
   where recall that $\tilde{\Sigma}^2 = \hat{\Sigma}_{d}^{2}\oplus\left(\gamma\hat{\Sigma}_{r}\right)^{2}$. It can similarly be shown that $\hat{V}^T\breve{V} = 0$, $\breve{V}\hat{V}^T = 0$, and $\breve{V}^T\breve{V}=I$.

   Next, we will show that $S=\hat{U}\hat{\Sigma}\hat{V}^T$.
   \begin{align*}
      \hat{U} \hat{\Sigma}\hat{V}^T &=
         \begin{bmatrix}\hat{U}_{d}\otimes \hat{U}_{r} & \hat{U}_{d}\otimes \breve{U}_{r} & \breve{U}_{d}\otimes \hat{U}_{r}\end{bmatrix}
            \begin{bmatrix}
         		\gamma\left(\hat{U}_{d}^T\otimes \hat{\Sigma}_{r}\hat{V}_{r}^T\right) &
               \left(\hat{\Sigma}_{d}\hat{V}_{d}^T\otimes \hat{U}_{r}^T\right) \\
               0 & \hat{\Sigma}_{d}\hat{V}_{d}^T\otimes \breve{U}_{r}^T \\
               \gamma\breve{U}_{d}^T\otimes \hat{\Sigma}_{r}\hat{V}_{r}^T & 0
         	\end{bmatrix} \\
      &= \begin{bmatrix}
         \gamma (\hat{U}_d \hat{U}_d^T + \breve{U}_d \breve{U}_d^T) \otimes S_r &
         S_d \otimes (\hat{U}_r \hat{U}_{r}^T + \breve{U}_r \breve{U}_{r}^T) \\
      \end{bmatrix} \\
      &= \begin{bmatrix}
        \gamma U_d U_d^T \otimes S_r &
        S_d \otimes U_r U_r^T \\
     \end{bmatrix} \\
      &= \begin{bmatrix}
         \gamma I_n \otimes S_r &
         S_d \otimes I_m \\
        \end{bmatrix}
   \end{align*}

	Therefore, Theorem~\ref{thm:svd-S} gives the SVD of $S$ at all values of $\gamma$.
\end{proof}

\subsection{Error analysis for example system}

In this section we present an error analysis for the approximate SVD of an example stoichiometry matrix. We demonstrate numerically that the approximate SVD presented in Theorem~\ref{thm:svd-S} converges to the true SVD in the limit as $\gamma \rightarrow 0$. We will consider a simplified set of equations that describes part of the Calvin Cycle in cyanobacteria. Specifically, cyanobacteria have cellular compartments called carboxysomes that serve to concentrate carbon within the cell \cite{Price2008}. It is thought that this compartmentalization increases the amount of carbon fixation and decreases the flux through the competing photorespiration reaction. This is an example of the type of system that could, in the future, be investigated more thoroughly with the approach presented here.

We consider a system with $n=8$ compartments, where $n_1=2$ and $n_2=6$. The first subregion in the domain represents the carboxysome and the second region represents the cytoplasm. We will consider the scenario of Mixed boundary conditions where fluxes are allowed only into the right side of the domain (i.e., into the cytoplasm region). Biologically, this scenario could represent a radially symmetric region in the cell centered on a carboxysome. The species in this system, as ordered in the stoichiometry matrix, are bicarbonate (\ce{HCO_3^-}), Ribulose 1,5-bisphosphate (\ce{RuBP}), carbon dioxide (\ce{CO_2}), 3-phosphogylcerate (\ce{3PGA}), Oxygen (\ce{O_2}), and 2-phosphoglycolate (\ce{2PG}). The reactions are given as
\begin{equation}\label{eq:ex-reactions}
    \begin{aligned}
        \text{(R1)} & & \ce{RuBP + CO_2 &-> 2 (3PGA)}  & \text{(Carbon Fixation})\\
        \text{(R2)} & & \ce{RuBP + O_2 &-> 3PGA + 2PG} & \text{(Photorespiration)}\\
        \text{(R3)} & & \ce{ CO_2 &<=> HCO_3^-} \\
    \end{aligned}
\end{equation}
It is known that \ce{O_2} and \ce{CO_2} cannot diffuse into the carboxysome \cite{Kinney2011,Mahinthichaichan2018}. Therefore we set $\mathcal{M}_- = \{3,5\}$ and $\mathcal{M}_+ = \{1,2,4,6\}$.

Given that \ce{O_2} is not present in the carboxysome, we know that only R1 and R3 occur in the first subregion. This leads to the following reaction-only stoichiometry matrices in the first and second region, respectively,
\begin{equation}\label{eq:Sr-ex}
    S_{r_1} =
        \begin{bmatrix}
            0 & -1 \\
            -1 & 0 \\
            -1 & 1 \\
            2 & 0 \\
            0 & 0 \\
            0 & 0
        \end{bmatrix}
   \quad \quad S_{r_2} =
      \begin{bmatrix}
           0 & 0 & -1 \\
           -1 & -1 & 0 \\
           -1 & 0 & 1 \\
           2 & 1 & 0 \\
           0 & -1 & 0 \\
           0 & 1 & 0
      \end{bmatrix}.
\end{equation}

Using the defined parameters, we applied the equations in Theorem~\ref{thm:svd-S} at multiple values of $\gamma$ and compared the results to the numerical SVD in MATLAB (Figure~\ref{fig}). As expected we find that the error approaches zero as $\gamma \rightarrow 0$. In this comparison the singular vectors/values are sorted by the magnitude of the singular value. Singular values between the numerical and approximate SVD (and hence singular vectors) are paired by finding those that are closest to each other in size. Note that in the error analysis in Figure~\ref{fig}, we only consider the nonzero singular values and corresponding singular vectors. Similar results are observed for the right and left null space (e.g., $S \check{V} \rightarrow 0$ as $\gamma \rightarrow 0$).

\begin{figure}
    \centering
    \includegraphics[scale=.8,trim=0 20 0 10, clip]{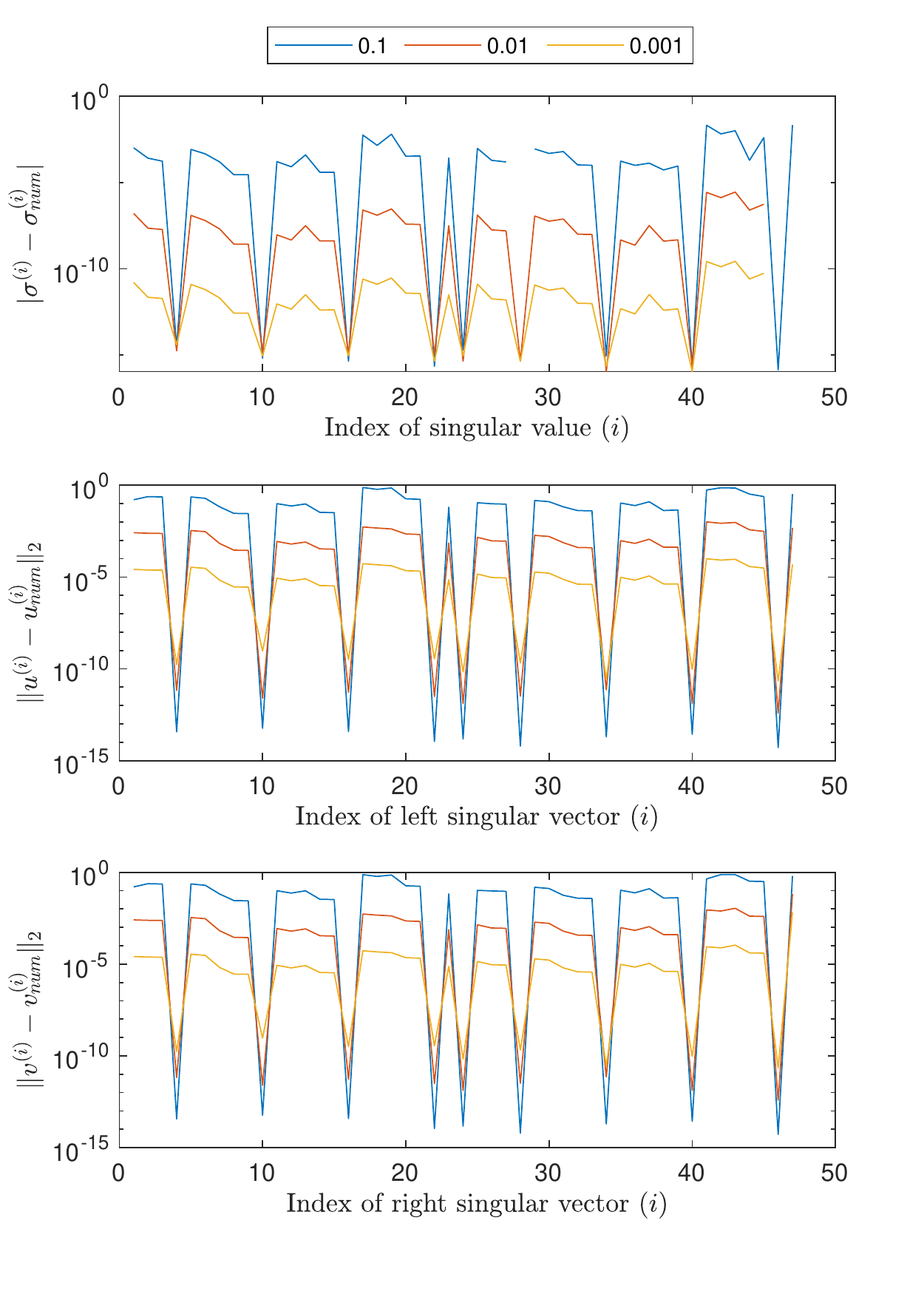}
    \caption{Error of approximate SVD at decreasing values of $\gamma$. The `$num$' subscript refers to singular vectors/values obtained numerically using MATLAB.}
    \label{fig}
\end{figure}

\section{Discussion}\label{sec:discussion}

In this paper we derived the approximate SVD of the stoichiometry matrix for a one dimensional discrete reaction-diffusion system partitioned into two subregions. Between these two subregions only certain species are allowed to diffuse. We additionally presented the exact SVD in the scenario where diffusion is allowed freely throughout the domain. This work provides a framework that can be applied and expanded upon to examine a variety of reaction-diffusion scenarios. For example, we hypothesize that in more complex scenarios (e.g., species-dependent boundary conditions) a Kronecker product formulation can still be used to write the SVD. Additionally, the formulas given by Theorem~\ref{thm:svd-S} and \ref{thm:svd-S-II} allow for future analysis looking at the effects of spatial properties in compartmentalized systems.

Computationally, the results of Theorem~\ref{thm:svd-S} and \ref{thm:svd-S-II} allow for the efficient estimation of the SVD of the RD stoichiometry matrix. Importantly, the approximate SVD is fully determined by the SVDs of smaller matrices. Either the SVD of these smaller matrices is known analytically or the dimension of the matrices is independent of the number of spatial compartments in the system. For example, consider a system with $m$ species, $p$ reactions, and either $n$ or $2n$ compartments. The diffusion-only stoichiometry matrices for this system are known analytically. The other matrices that the SVD depends on have dimensions proportional to $m$ and/or $p$. Notably the total number of required smaller matrix decompositions will increase linearly with the number of compartments.

\subsection{Intuition for SVD results}

The approximate SVD for a system with a barrier provides intuition for how the system's structure influences dynamical and steady state properties. As written in Theorem~\ref{thm:svd-S}, we have partitioned the singular vectors  into multiple sets, which we will refer to as \textit{eigenreaction sets}. For example, the singular vectors in $\hat{U}_{1,j}$ and $\hat{V}_{1,j}$ and the singular values in $\hat{\Sigma}_{1,j}$ for $j \in \mathcal{J}$ represent the first eigenreaction set.  Recall that the SVD defines eigenreactions, which represent decoupled linear combinations of species that are moved by a linear combinations of fluxes, e.g.,
\begin{equation}
   \frac{d}{dt} \left(u^{(i)}\right)^T w = \sigma_i \left(v^{(i)}\right)^T f.
\end{equation}
where $w$ is the species concentration vector and $f$ is the vector of fluxes, e.g., see (\ref{eq:MassBalance}).

Each eigenreaction set describes the movement of species with similar diffusive properties. The first eigenreaction set describes the movement of species that are able to diffuse across the barrier. The second, third and fourth eigenreaction sets all describe the movement of species that are unable to diffuse across the barrier. The second eigenreaction set includes both reactive and diffusive movement in the first subregion, the third eigenreaction set describes only reactive movement in the first subregion, and, finally, the fourth eigenreaction set describes reactive and diffusive movement in the second subregion. Recall that the third eigenreaction set is only nonempty for Mixed boundary conditions.

The fifth and sixth eigenreaction sets are unique in that they describe the movement of all species in the system. This movement is coupled due to the repeating singular values in the diffusion-only stoichiometry matrices (i.e., $S_d$, $S_{d_1}$, and $S_{d_2}$). This demonstrates that even in the regime where diffusion is much faster than reactions, there is still a coupling between species with different diffusive processes.

The basis for the nullspace of $S$ is also partition into multiple sets, given by Theorem~\ref{thm:svd-S} (i.e., $\breve{V}_i$ for $i=1,...,5$). We will refer to these as \textit{steady-state flux sets}, since they represent fluxes that can exist under steady-state conditions. The first and third steady-state flux sets include only reactive fluxes in the first subregion and throughout the domain, respectively. Note that the first steady-state flux set is only nonempty for Mixed boundary conditions, whereas the third steady-state flux set is only nonempty for Zero Flux boundary conditions. The second and fourth steady-state flux sets represent reactive and diffusive flux combinations in the first subregion and second subregion, respectively. Finally the fifth steady-state flux set represents fluxes that are in the nullspace of $S$ due to their infeasibility. That is, for a given dynamical system, these fluxes will never contribute since they define fluxes across barriers/boundaries that are not allowed.

\subsection{Conclusion and Future Work}

To find the SVD of the RD stoichiometry matrix, we first used linear perturbation theory to calculate the left singular vectors and values. We then used the resulting vectors and values to find the right singular vectors. An alternative approach would be to instead derive the right singular vectors directly using perturbation theory. Although this approach may provide additional insight into the system properties, it is slightly more complex as it involves additional terms in the expansions used in the perturbation analysis. Therefore, this analysis is the topic of future research.

The key assumption used to derive the approximate SVD is that diffusion is much faster than the reactions. Whether this is a valid assumption depends on the specific biological system under consideration. Indeed, the relative time-scales of diffusion and reactions in biological systems can very greatly and is a complex topic \cite{Soh2010}. A similar approach, as presented in this paper, could be applied to derive the approximate SVD in a system where reactions occur much faster than diffusion. That is, we would instead consider the perturbation problem in the limit as the diffusive term of (\ref{eq:SbSbT-eig}) goes to zero. Rigorously showing whether the approach applied here could work in this alternative case is a topic of future research.

Our motivation in deriving the approximate SVD in terms of reduced systems is to gain insight into how including spatial barriers and diffusion impact a biological system. The SVD for the reaction-only system has provided valuable insight in comparing genome-scale metabolic networks \cite{Famili2003} and finding connections between biochemical processes \cite{Palese2012}. By including spatial parameters, the work presented here provides tools that computational biologists can apply to understand how reactive processes are coupled across space.

\section*{Acknowledgments}

We would like to thank Professor Jeffrey Cameron at CU Boulder for insightful discussions about this work and potential biological applications.

\bibliographystyle{siamplain}
\bibliography{/home/dmbortz/Desktop/SVD_Stoic_Mat_18April2021/bibfile}

\begin{thebibliography}{10}

\bibitem{Arcak2011}
{\sc M.~Arcak}, {\em {Certifying spatially uniform behavior in
  reactiondiffusion PDE and compartmental ODE systems}}, Automatica, 47 (2011),
  pp.~1219--1229, \url{https://doi.org/10.1016/j.automatica.2011.01.010}.

\bibitem{Clarke1988}
{\sc B.~L. Clarke}, {\em {Stoichiometric network analysis}}, Cell Biophysics,
  12 (1988), pp.~237--253, \url{https://doi.org/10.1007/BF02918360}.

\bibitem{DAutilia2020}
{\sc M.~C. D'Autilia, I.~Sgura, and V.~Simoncini}, {\em {Matrix-oriented
  discretization methods for reaction–diffusion PDEs: Comparisons and
  applications}}, Computers and Mathematics with Applications, 79 (2020),
  pp.~2067--2085, \url{https://doi.org/10.1016/j.camwa.2019.10.020},
  \url{https://arxiv.org/abs/1903.05030}.

\bibitem{Famili2003}
{\sc I.~Famili and B.~O. Palsson}, {\em {Systemic metabolic reactions are
  obtained by singular value decomposition of genome-scale stoichiometric
  matrices}}, Journal of Theoretical Biology, 224 (2003), pp.~87--96,
  \url{https://doi.org/10.1016/S0022-5193(03)00146-2},
  \url{http://www.sontaglab.org/sysbio{\_}papers{\_}readings/famili-palsson-metabolic-networks-svd-JTB03.pdf}.

\bibitem{Johnson2011}
{\sc K.~A. Johnson and R.~S. Goody}, {\em {The original Michaelis constant:
  Translation of the 1913 Michaelis-Menten paper}}, Biochemistry, 50 (2011),
  pp.~8264--8269, \url{https://doi.org/10.1021/bi201284u}.

\bibitem{Kato1995}
{\sc T.~Kato}, {\em {Perturbation Theory for Linear Operators}}, vol.~132 of
  Classics in Mathematics, Springer Berlin Heidelberg, Berlin, Heidelberg,
  1995, \url{https://doi.org/10.1007/978-3-642-66282-9},
  \url{https://books.google.com/books?hl=en{\&}lr={\&}id=k-7nCAAAQBAJ{\&}oi=fnd{\&}pg=PR17{\&}dq=kato+perturbation+theory{\&}ots=w6PFpfgaxC{\&}sig=2ex8pvcT4wZyZ1pBfbHQnJoSeVg
  http://link.springer.com/10.1007/978-3-642-66282-9}.

\bibitem{Kinney2011}
{\sc J.~N. Kinney, S.~D. Axen, and C.~A. Kerfeld}, {\em Comparative analysis of
  carboxysome shell proteins}, Photosynthesis Research, 109 (2011), pp.~21--32.

\bibitem{Loan2000}
{\sc C.~F. Loan}, {\em {The ubiquitous Kronecker product}}, Journal of
  Computational and Applied Mathematics, 123 (2000), pp.~85--100,
  \url{https://doi.org/10.1016/S0377-0427(00)00393-9}.

\bibitem{Mahinthichaichan2018}
{\sc P.~Mahinthichaichan, D.~M. Morris, Y.~Wang, G.~J. Jensen, and
  E.~Tajkhorshid}, {\em Selective permeability of carboxysome shell pores to
  anionic molecules}, The Journal of Physical Chemistry B, 122 (2018),
  pp.~9110--9118.

\bibitem{Palese2012}
{\sc L.~L. Palese and F.~Bossis}, {\em {The human extended mitochondrial
  metabolic network: New hubs from lipids}}, BioSystems, 109 (2012),
  pp.~151--158, \url{https://doi.org/10.1016/j.biosystems.2012.04.001},
  \url{https://www.sciencedirect.com/science/article/pii/S030326471200055X}.

\bibitem{Palsson2006}
{\sc B.~Palsson}, {\em {Systems biology: properties of reconstructed
  networks}}, Cambridge and New York: Cambridge University Press, 2006.

\bibitem{Price2008}
{\sc G.~D. Price, M.~R. Badger, F.~J. Woodger, and B.~M. Long}, {\em {Advances
  in understanding the cyanobacterial CO2-concentrating- mechanism (CCM):
  Functional components, Ci transporters, diversity, genetic regulation and
  prospects for engineering into plants}}, in Journal of Experimental Botany,
  vol.~59, 2008, pp.~1441--1461, \url{https://doi.org/10.1093/jxb/erm112},
  \url{https://academic.oup.com/jxb/article-abstract/59/7/1441/636484}.

\bibitem{Soh2010}
{\sc S.~Soh, M.~Byrska, K.~Kandere-Grzybowska, and B.~A. Grzybowski}, {\em
  {Reaction-diffusion systems in intracellular molecular transport and
  control}}, 2010, \url{https://doi.org/10.1002/anie.200905513},
  \url{www.dysa.northwestern.edu}.

\bibitem{Voit2015}
{\sc E.~O. Voit, H.~A. Martens, and S.~W. Omholt}, {\em {150 Years of the Mass
  Action Law}}, PLoS Computational Biology, 11 (2015), p.~e1004012,
  \url{https://doi.org/10.1371/journal.pcbi.1004012}.

\bibitem{Wentz2020}
{\sc J.~M. Wentz and D.~M. Bortz}, {\em {Boundedness of a class of discretized
  reaction-diffusion systems}}, arXiv (submitted to SIAP),  (2020),
  \url{http://arxiv.org/abs/1903.09680},
  \url{https://arxiv.org/abs/1903.09680}.

\end{thebibliography}

\appendix
\section{Background on linear perturbation theory}\label{sec-app:perturbation-theory}

Here, we present background information on concepts from linear perturbation theory that is used to derive the approximate singular value decomposition (SVD) of the stoichiometry matrix for the reaction-diffusion system with a barrier. We refer the reader to \cite{Kato1995} for a more thorough description of this material. Our discussion here focuses on symmetric martrices. This allows us to assume that the eigenvalues are semisimple and, therefore, the eigennilopotents (denoted with a $D$ in \cite{Kato1995}) vanish.

 Consider the following matrix
\begin{equation}
	T(x) = T + x T^{(1)}
\end{equation}
where $T \in \mathbb{R}^{n\times n}$ and $T^{(1)} \in \mathbb{R}^{n\times n}$ are symmetric matrices and $x\ge0$. We will refer to $T$ as the unperturbed matrix to $T^{(1)}$ as the perturbation matrix.

Our goal is to find an approximate eigendecomposition of $T(x)$ at small $x$. Consider the following eigenvalue problem
\begin{equation}
	T(x)u_i(x) = \lambda_i(x) u_i(x).
\end{equation}
Additionally, $\lambda_i(x)$ is a continuous function of $x$ (see Theorem 2.3 from \cite{Kato1995}, Chapter 2, Section 2.3), implying that as $x\rightarrow 0$ the eigenvalues of $T(x)$ are equal to the eigenvalues of $T$. However, the same statement does not hold for the eigenvectors. That is, suppose there exists $j \ne i$ such that $\lambda_j(0) = \lambda_i(0)$, but for arbitrarily small $x>0$ $\lambda_k(x) \ne \lambda_i(x)$ for all $k \ne i$. In this scenario, the eigenvector that corresponds to $\lambda_i(x)$ is unique, but the eigenvector that corresponds to $\lambda_i(0)$ is not. Our task is to find the `correct' set of eigenvectors such that $u_i(x)$ converges to $u_i(0)$ as $x \rightarrow 0$.

More generally, let $\lambda$ be an eigenvalue of $T$ with multiplicity $m$ and denote the $m$ eigenvalues such that $\lambda_k(0) = \lambda$ as the $\lambda$-group. Without loss of generality, suppose this is the first $m$ eigenvalues. Let $(u_k)_{i=1}^m$ represent a set of orthogonal eigenvectors that solve the eigenvalue problem $Tu_k(0)=\lambda u_k(0)$. Let
\begin{equation}\label{eq:proj}
	P = \sum_{k=1}^m u_k u_k^T
\end{equation} be the unique orthogonal eigenprojection associated with $\lambda$ (i.e., $P^2=P$ and $TP = \lambda P$). We will additionally consider the sum of projections at small $x$ for the entire $\lambda$-group
\begin{equation}\label{eq:sumProj}
	P(x) = \sum_{k=1}^m u_k(x) u_k(x)^T.
\end{equation}
Since, in practice, it is difficult to find $u_k(x)$, we can instead write $P(x)$ using a contour integral of the resolvent. That is, let the resolvent of $T(x)$ at the point $\eta$ be given as
\begin{equation}\label{eq:R}
	R(\zeta,x)=(T(x)-\eta I)^{-1}
\end{equation}
and let $\Gamma$ be a closed positively-oriented curve in the resolvent set that encloses $\lambda$ and no other eigenvalues of $T$. The projection
\begin{equation}\label{eq:projInt}
	P(x) = -\frac{1}{2\pi i} \int_{\Gamma} R(\zeta,x)d\zeta
\end{equation}
is equal to the sum of the eigenprojections for eigenvalues of $T(x)$ that lie inside $\Gamma$ (see \cite{Kato1995}, Chapter 2, Section 1.4).

To find the eigendecomposition of $T(x)$ at small $x$, we will instead consider the equivalent eigenvalue problem for
\begin{equation}\label{eq:Ttilde}
	\tilde{T}^{(1)}(x) = \frac{1}{x}(T(x)-\lambda I)P(x).
\end{equation}
as $x \rightarrow 0$. To see that these eigenvalue problems are equivalent, first note that since the eigenvalues $\lambda_k(x)$ are continuously differentiable in a neighborhood of $x=0$ (see Theorem 2.3 from \cite{Kato1995}, Chapter 2, Section 2.3), we can write the power series expansion of $\lambda_k(x)$ as
\begin{equation}\label{eq:lambdakx}
   \lambda_k(x) = \lambda + x\lambda_k^{(1)} + \mathcal{O}(x^2).
\end{equation}
Then, using (\ref{eq:lambdakx}) and $P(x)u_k(x)=u_k(x)$ for $k=1,...,m$ we obtain
\begin{align*}
	T(x)u_k(x) &= \lambda_k(x) u_k(x) \\
	T(x)P(x)u_k(x) &= (\lambda + x\lambda_k^{(1)} +\mathcal{O}(x^2))u_k(x) \\
	(T(x)-\lambda I)P(x)u_k(x) &= (x\lambda_k^{(1)} +\mathcal{O}(x^2))u_k(x) \\
	\frac{1}{x}(T(x)-\lambda I)P(x)u_k(x) &= (\lambda_k^{(1)} +\mathcal{O}(x))u_k(x).
\end{align*}
Therefore, the eigenvectors of $T(x)$ are equal to the eigenvectors of $\tilde{T}^{(1)}(x)$, and the associated eigenvalues of $T(x)$ can be written as
\begin{equation}
	\lambda_k(x) = \lambda + x\lambda_{1,k} +\mathcal{O}(x^2),
\end{equation}
where $\lambda_{1,k}=\lambda_k^{(1)}+\mathcal{O}(x)$ is the eigenvalue of $\tilde{T}^{(1)}(x)$ associated with eigenvector $u_k(x)$.

Next, we can use power series expansions to show that
\begin{equation}
	\tilde{T}^{(1)}(x) = \tilde{T}^{(1)} + \mathcal{O}(x).
\end{equation}
   First, note the resolvent can be written as
\begin{equation}\label{eq:resExp}
	R(\zeta,x) = R(\zeta) - x R(\zeta)T^{(1)}R(\zeta) + \mathcal{O}(x^2)
\end{equation}
where $R(\zeta) = R(\zeta,0)$ (see Chapter 2, Section 1.3 of \cite{Kato1995} for derivation). Using (\ref{eq:resExp}), we can write the sum of eigenprojections for the $\lambda$-group as
\begin{equation}
	P(x) = P + x P^{(1)} + \mathcal{O}(x^2) \quad \text{ where } \quad P^{(1)} = -\frac{1}{2\pi i} \int_{\Gamma} R(\zeta)T^{(1)}R(\zeta) d \zeta.
\end{equation}

Using (\ref{eq:R}), (\ref{eq:projInt}), (\ref{eq:resExp}), and $(T-\lambda I)P=0$ we have that
\begin{equation}
	\begin{aligned}
		(T(x)-\lambda I)P(x) &= -\frac{1}{2\pi i} (T(x)-\lambda I)\int_{\Gamma} R(\zeta,x)d\zeta \\
      &= -\frac{1}{2\pi i} \int_{\Gamma} (T(x)-\lambda I) (T(x)-\eta I)^{-1} d\zeta \\
      &= - \frac{1}{2\pi i} \int_{\Gamma} I + (\zeta-\lambda) R(\zeta,x)d\zeta \\
      &= - \frac{1}{2\pi i} \int_{\Gamma}  (\zeta-\lambda) (R(\zeta) - x R(\zeta)T^{(1)}R(\zeta))d\zeta + \mathcal{O}(x^2) \\
      &= (T-\lambda I)P + \frac{1}{2\pi i} \int_{\Gamma}  (\zeta-\lambda) (x R(\zeta)T^{(1)}R(\zeta))d\zeta + \mathcal{O}(x^2) \\
		&= x \tilde{T}^{(1)} + \mathcal{O}(x^2)
	\end{aligned}
\end{equation}
where
\begin{equation}
	\tilde{T}^{(1)} = \frac{1}{2\pi i} \int_{\Gamma} R(\zeta) T^{(1)} R(\zeta)(\zeta-\lambda) d\zeta.
\end{equation}
We can evaluate this integral by substituting $R(\zeta)$ by its Laurent expansion at $\zeta=\lambda$, i.e.,
\begin{equation}
	R(\zeta) = \sum_{n=-1}^{\infty} (\zeta-\lambda)^n S^{(n+1)}
\end{equation}
where
\begin{equation}
   S^{(0)}=-P, \quad S^{(n)} = S^n
\end{equation}
where $S=S(\lambda)$ is the value at $\zeta=\lambda$ of the reduced resolvent of $T$. Using the Cauchy residue theorem

\begin{equation}
   \begin{aligned}
   \tilde{T}^{(1)} &= \frac{1}{2\pi i} \int_{\Gamma} (\zeta-\lambda) \left(\sum_{n=-1}^{\infty} (\zeta-\lambda)^n S^{(n+1)}\right) T^{(1)} \left(\sum_{n=-1}^{\infty} (\zeta-\lambda)^n S^{(n+1)}\right)d\zeta \\
   &= \frac{1}{2\pi i} \int_{\Gamma} (\zeta-\lambda) \left((\zeta-\lambda)^{-1} S^{(0)}\right) T^{(1)} \left((\zeta-\lambda)^{-1} S^{(0)}\right)d\zeta \\
   &= \frac{1}{2\pi i} \int_{\Gamma} (\zeta-\lambda)^{-1} P T^{(1)}  P\ d\zeta \\
   &= P T^{(1)}  P \\
   \end{aligned}
\end{equation}
Notice that terms with $(\zeta - \lambda)^n$ where $n>0$ in the integral vanish since there is no singularity.

Putting these results together, if $\lambda$ is an eigenvalue of $T(0)$ that repeats $m$ times, then at small $x$ the associated eigenvalues and eigenvectors of $T(x)$ can be approximated, for $k=1,...,m$ as
\begin{align}
   u_k(x) &= \tilde{u}_k(x)\\
   \lambda_k(x) &= \lambda + x\tilde{\lambda}_k\
\end{align}
where $\tilde{u}_k(x)$ are the eigenvectors of $\tilde{T}^{(1)}$ that are in the range of $P$ and $\tilde{\lambda}_k$ are the corresponding eigenvalues of $\tilde{T}^{(1)}$.

\section{Singular value decomposition of $S_d$}
\label{sec-app:SVD-Sd}

In this section we provide the explicit SVD of the diffusion-only stoichiometry matrix, $S_d$, as given by (\ref{eq:Sd}). We will consider a system with $n$ compartments, but note that by replacing $n$ with $n_1$ or $n_2$ this notation can be used to define the SVD of the diffusion-only stoichiometry matrices for the two subregions, $S_{d_1}$ and $S_{d_2}$.

In the main manuscript we present three possible boundary conditions: Zero Flux, Mixed, and Open. We will additionally included formulas for what we call \textit{Mixed-Alt} boundary conditions, which can be thought of as the opposite of Mixed boundary conditions (i.e., where input/output flux is allowed at $x=0$ but not at $x=n$). We include this additional boundary condition because it is used to describe the first subregion in a system with a barrier and Dirichlet boundary conditions.

The SVD will depend on the following constants for $j=1,...,n$,
\begin{align}
a_{n,j} &= \frac{\pi(n-j)}{2n} \\
b_{n,j} &= \frac{\pi(n-j+\nicefrac{1}{2})}{2n+1} \\
c_{n,j} &= \frac{\pi (n-j+1)}{2(n+1)}.
\end{align}
Next we define the left singular vectors that correspond to the column space and left nullspace. For the left singular vectors in the column space, the $i$th element of the $j$th left singular vector is, for $i=1,...,n$ and $j=1,...,q_d$,
\begin{equation}
\left(u_{d}^{(j)}\right)_i=
\begin{cases}
\sqrt{\frac{2}{n}}\cos\left(2a_{n,j}(i-\nicefrac{1}{2})\right) & \text{Zero Flux}\\
\sqrt{\frac{2}{n+\nicefrac{1}{2}}}
\cos\left(2b_{n,j}(i-\nicefrac{1}{2})\right) &
\text{Mixed} \\
\sqrt{\frac{2}{n+\nicefrac{1}{2}}}
\sin\left(2b_{n,j}i\right) &
\text{Mixed-Alt} \\
\sqrt{\frac{2}{n+1}}\sin\left(2c_{n,j}i\right) &
\text{Open}.
\end{cases}
\end{equation}
The left nullspace is only nonempty for Zero flux boundary conditions and we have that
\begin{equation}
   \left(u_{d}^{(n)}\right)_i = \frac{1}{\sqrt{n}}.
\end{equation}

For the right singular vectors, the $i$th element of the $j$th right singular vector associated with nonzero singular values is, for $i=1,...,n+1$ and $j=1,...,q_d$,
\begin{equation}
\left(v_{d}^{(j)}\right)_i=
\begin{cases}
-\sqrt{\frac{2}{n}}\sin\left(2a_{n,j}(i-1)\right) & \text{Zero Flux} \\
-\sqrt{\frac{2}{n+\nicefrac{1}{2}}}\sin\left(2 b_{n,j} (i-1) \right)
& \text{Mixed} \\
\sqrt{\frac{2}{n+\nicefrac{1}{2}}}\cos\left(2 b_{n,j} (i-\nicefrac{1}{2})\right)
& \text{Mixed Alt} \\
\sqrt{\frac{2}{n+1}}\cos\left(2 c_{n,j} (i-\nicefrac{1}{2})\right) & \text{Open}.
\end{cases}
\end{equation}
For the right singular vectors in the nullspace of $S_d$, for Zero Flux boundary conditions, we have that
\begin{align}
   v_{d}^{(n)} = e_1, \quad \quad
   v_{d}^{(n+1)} = e_{n+1}\
\end{align}
where $e_i$ represents the vector with zeros and a one at the $i$th index.
For Mixed, Mixed-Alt, and Open boundary conditions, we have
\begin{align}
   v_d^{(n+1)} = \begin{cases}
      e_1 & \text{Mixed} \\
      e_{n+1} & \text{Mixed-Alt} \\
      \bm{1}\sqrt{\frac{1}{n+1}} & \text{Open}
   \end{cases}
\end{align}
where $\bm{1}$ is a vector of ones.

Finally, the $j$th singular value, for each of the boundary conditions, is
\begin{equation}
\sigma_{d}^{(j)}=
\begin{cases}
2 \sin\left(a_{n,j}\right) & \text{Zero Flux} \\
2 \sin\left(b_{n,j}\right) & \text{Mixed and Mixed-Alt} \\
2 \sin\left(c_{n,j}\right) & \text{Open}. \\
\end{cases}
\label{app-eq:sigmad}
\end{equation}

\section{Kronecker product formulas}\label{sec-app:Kronecker-formulas}

In this section we provide some Kronecker product relations that are needed to prove Theorem \ref{thm:svd-S}. We omit the proof of these properties but note that they can be shown through a series of algebraic manipulations.

\begin{property}\label{app-property:Kron1}
	Let $u \in \mathbb{R}^{n \times 1}$ be related to $u_1 \in \mathbb{R}^{n_1 \times 1}$ and $u_2 \in \mathbb{R}^{n_2 \times 2}$ such that
	\begin{equation}
	u = \begin{bmatrix}
	u_1 \\ u_2
	\end{bmatrix}.
	\end{equation}
	Let $A_1,A_2,B_1,B_2$ be square matrices of the same size. Then
	\begin{equation}
	\left(uu^T \otimes A_1\right)\begin{bmatrix}
	I_{n_1} \otimes B_1 & 0 \\
	0 & I_{n_2} \otimes B_2 \\
	\end{bmatrix} \left(u u^T \otimes A_2\right) =
	uu^T \otimes A_1 \left(|u_1|^2 B_1 + |u_2|^2 B_2 \right) A_2.
	\end{equation}
\end{property}

\begin{property}\label{app-property:Kron2}
	Suppose $u_1$ and $u_2$ are unit vectors and $a_i \in \mathbb{R}$ for $i=1,..,8$. Let $A_1,A_2,B_1,B_2$ be square matrices of the same size. Then
	\begin{multline}
	\left(\begin{bmatrix}
	a_1 u_1 u_1^T & a_2 u_1 u_2^T \\
	a_3 u_2 u_1^T & a_4 u_2 u_2^T
	\end{bmatrix} \otimes A_1 \right)
	\begin{bmatrix}
	I_{n_1} \otimes B_1 & 0 \\
	0 & I_{n_2} \otimes B_2 \\
	\end{bmatrix}
	\left(\begin{bmatrix}
	a_5 u_1 u_1^T & a_6 u_1 u_2^T \\
	a_7 u_2 u_1^T & a_8 u_2 u_2^T
	\end{bmatrix} \otimes A_2 \right)\\
	=
	\begin{cases} \vspace{3pt}
	\begin{bmatrix}
	a_5 u_1 u_1^T &
	a_6 u_1 u_2^T \\
	0 &
	0
	\end{bmatrix} \otimes A_1 \left(a_1 B_1\right) A_2, & \text{ if }a_2,a_3,a_4=0 \\ \vspace{3pt}
	\begin{bmatrix}
	0 &
	0 \\
	a_7 u_2 u_1^T &
	a_8 u_2 u_2^T
	\end{bmatrix} \otimes A_1 \left(a_4 B_2 \right) A_2, & \text{ if }a_1,a_2,a_3=0  \\ \vspace{3pt}
	\begin{bmatrix}
	a_1 u_1 u_1^T  &
	0 \\
	a_3 u_2 u_1^T &
	0
	\end{bmatrix} \otimes A_1 \left(a_5 B_1 \right) A_2, & \text{ if }a_6,a_7,a_8=0  \\
	\begin{bmatrix}
	0 &
	a_2 u_1 u_2^T  \\
	0 &
	a_4 u_2 u_2^T
	\end{bmatrix} \otimes A_1 \left(a_8 B_2 \right) A_2, & \text{ if }a_5,a_6,a_7=0.
	\end{cases}
	\end{multline}
\end{property}

\begin{property}\label{app-property:Kron3}
	Suppose that $u_1$ and $u_2$ are unit column vectors and $a_i \in \mathbb{R}$ for $i=1,..,6$ such that
	\begin{equation}
	  a_1 = \frac{a_3 a_5}{a_6} \quad \text{and} \quad
	  a_4 = \frac{a_2 a_6}{a_5}.
	\end{equation}
	Let $A_1 \in \mathbb{R}^{m\times n}$ and $v_1,v_2 \in \mathbb{R}^{m \times 1}$.	We then have that
	\begin{equation}
   	\begin{bmatrix}
   	a_1 u_1 u_1^T \otimes A_1 & a_2 u_1 u_2^T \otimes A_1 \\
   	a_3 u_2 u_1^T \otimes A_1 & a_4 u_2 u_2^T \otimes A_1
   	\end{bmatrix}
   	\begin{bmatrix}
   	a_5 u_1 \otimes v_1 \\
   	a_6 u_2 \otimes v_2
   	\end{bmatrix}
   	=
   	\begin{bmatrix}
   	a_5 u_1\\
   	a_6 u_2
   	\end{bmatrix} \otimes A_1\left(a_1 v_1 + a_4 v_2\right).
	\end{equation}
\end{property}

\section{Proofs}\label{sec-app:proofs}

This section contains supplemental proofs for the results presented in Section~\ref{sec:useful-results} and \ref{sec:SVD}. We first provide the proof of Lemma~\ref{lemma:eig-repeats}, which provides a relationship for the eigenvalues and eigenvectors of $S_d$, $S_{d_1}$ and $S_{d_2}$.

\begin{proof}[Proof of Lemma~\ref{lemma:eig-repeats}]

   We will consider the three possible boundary conditions independently.

   \begin{description}\item[Case 1: Homogeneous Neumann boundary conditions.]
		In this scenario, both subregions have homogeneous Neumann boundary conditions. Suppose that, for $j \in \{1,2,...,n\}$, there exists $j_1 \in \{1,2,...,n_1\}$ such that
		\begin{equation}
      \begin{aligned}
   		a_{n,j} &= \frac{\pi(n-j)}{2n} = \frac{\pi(n_1-j_1)}{2n_1} = a_{n_1,j_1} \\
   		&\implies j_1 = \frac{n_1j}{n} = C_1^2 j.
   		\label{eq:j1-Neumann}
      \end{aligned}
		\end{equation}
		From (\ref{app-eq:sigmad}) this implies that $\sigma_{d_1}^{(j_1)}=\sigma^{(j)}$.  Let $j_2 = j-j_1 = C_2^2 j$ and note that by definition $j_2 \in \{1,...,n_2\}$. Additionally, $\sigma^{(j_2)}_{d_2}=\sigma^{(j)}$ since
		\begin{equation}
		a_{n,j} = \frac{\pi(n-j)}{2n} = \frac{\pi(n_2- C_2^2 j)}{2n_2} = \frac{\pi(n_2-j_2)}{2n_2} = a_{n_2,j_2}.
		\label{eq:j2-Neumann}
		\end{equation}

		Next we will prove (\ref{eq:Ud_divided}) and (\ref{eq:Vd_divided}). In what follows let $\ell=i-n_1$. For $i=1,...,n_1$, (\ref{eq:j1-Neumann}) and (\ref{eq:j2-Neumann}) imply that
		\begin{align*}
   		\left(u_{d}^{(j)}\right)_i
   		&= \sqrt{\frac{2}{n}}\cos\left(2 a_{n,j}(i-\nicefrac{1}{2})\right)
   		= \sqrt{\frac{2}{n}}\cos\left(2 a_{n_1,j_1}(i-\nicefrac{1}{2})\right)
   		= C_1 \left(u_{d_1}^{(j)}\right)_i \\
   		\left(v_{d}^{(j)}\right)_i
   		&= \sqrt{\frac{2}{n}}\sin\left(2 a_{n,j}(i-1)\right)
   		= \sqrt{\frac{2}{n}}\sin\left(2 a_{n_1,j_1}(i-1)\right)
   		= C_1 \left(v_{d_1}^{(j_1)}\right)_i.
		\end{align*}

		For $i=n_1+1,...,n$, (\ref{eq:j1-Neumann}) and (\ref{eq:j2-Neumann}) imply that
      \begin{equation*}
		\begin{aligned}
			\left(u_{d}^{(j)}\right)_{n_1+\ell}
			&= \sqrt{\frac{2}{n}}\cos\left(2a_{n,j}\left(n_1+\ell-\frac{1}{2}\right)\right) \\
			&= \sqrt{\frac{2}{n}}\cos\left(\frac{2\pi(n_1-j_1)n_1}{2n_1}+2a_{n,j}\left(\ell-\frac{1}{2}\right)\right) \\
			&= (-1)^{n_1-j_1} \sqrt{\frac{2}{n}}\cos\left(2 a_{n_2,j_2}\left(\ell-\frac{1}{2}\right)\right) \\
			&= (-1)^{n_1-j_1} C_2 \left(u_{d_2}^{(j_2)}\right)_\ell
		\end{aligned}
      \end{equation*}
		and
		\begin{equation*}
      \begin{aligned}
			\left(v_{d}^{(j)}\right)_{n_1+\ell}
			&= \sqrt{\frac{2}{n}}\sin\left(2a_{n,j} (n_1+\ell-1)\right) \\
			&= \sqrt{\frac{2}{n}}\sin\left(\frac{2\pi(n_1-j_1)n_1}{2n_1}+2a_{n,j}(\ell-1)\right) \\
			&= (-1)^{n_1-j_1} \sqrt{\frac{2}{n}}\sin\left(2a_{n_2,j_2}(\ell-1)\right) \\
			&= (-1)^{n_1-j_1} C_2 \left(v_{d_2}^{(j_2)}\right)_\ell.
      \end{aligned}
      \end{equation*}
		\item[Case 2: Mixed boundary conditions.]
		In this scenario, the first subregion with $n_1$ compartments has homogeneous Neumann boundary conditions and the second subregion with $n_2$ compartment has Mixed boundary conditions. For $j\in \{1,...,n\}$ suppose there exists a $j_1 \in \{1,...,n_1\}$ such that
		\begin{equation}
      \begin{aligned}
         \label{eq:j1-Mixed}
   		b_{n,j} &= \frac{\pi(n-j+\nicefrac{1}{2})}{2n+1} = \frac{\pi(n_1-j_1)}{2n_1} = a_{n_1,j_1} \\
   		&\implies j_1 = \frac{2n_1j}{2n+1} = C_1^2 j.
      \end{aligned}
      \end{equation}
		From (\ref{app-eq:sigmad}) this implies that $\sigma_{d_1}^{(j_1)}=\sigma^{(j)}$. Let $j_2 = j-j_1 = C_2^2 j$ and note by definition that $j_2 \in \{1,..,n_2\}$. Additionally, $\sigma_{d_2}^{(j_2)} = \sigma^{(j)}$, since
		\begin{equation}
      \begin{aligned}\label{eq:j2-Mixed}
   		b_{n,j} &= \frac{\pi(n-j+\nicefrac{1}{2})}{2n+1} \\ &= \frac{\pi(n_2+\nicefrac{1}{2}-C_2^2 j)}{2n_2+1} \\
            &=	\frac{\pi(n_2-j_2+\nicefrac{1}{2})}{2n_2+1} = b_{n_2,j_2}
      \end{aligned}
      \end{equation}

		Next we will prove (\ref{eq:Ud_divided}) and (\ref{eq:Vd_divided}). For $i=1,...,n_1$, (\ref{eq:j1-Mixed}) implies that
		\begin{align*}
   		\left(u_{d}^{(j)}\right)_i
   		&= \sqrt{\frac{2}{n+\nicefrac{1}{2}}}
   		\cos\left(2b_{n,j} (i-\nicefrac{1}{2})\right) \\
   		&= \sqrt{\frac{2}{n+\nicefrac{1}{2}}}\cos\left(2 a_{n_1,j_1}(i-\nicefrac{1}{2})\right)
   		= C_1 \left(u_{d,n_1}^{(j_1)}\right)_i \\
   		\left(v_{d}^{(j)}\right)_i
   		&= \sqrt{\frac{2}{n+\nicefrac{1}{2}}}
   		\sin\left(2b_{n,j}(i-1)\right) \\
   		&= \sqrt{\frac{2}{n+\nicefrac{1}{2}}}\sin\left(2a_{n_1,j_1}(i-1)\right)
   		= C_1 \left(v_{d_1}^{(j_1)}\right)_i.
		\end{align*}
		For $i=n_1+1,...,n$, (\ref{eq:j1-Mixed}) and (\ref{eq:j2-Mixed}) imply that
      \begin{equation*}
		\begin{aligned}
			\left(u_{d}^{(j)}\right)_{n_1+\ell}
			&= \sqrt{\frac{2}{n+\nicefrac{1}{2}}}
			\cos\left(2b_{n,j}(n_1+\ell-\nicefrac{1}{2})\right) \\
			&= \sqrt{\frac{2}{n+\nicefrac{1}{2}}}
			\cos\left(\frac{2 \pi (n_1-j_1)n_1}{2n_1} +2b_{n,j}(\ell-\nicefrac{1}{2})\right) \\
			&= (-1)^{n_1-j_1}\sqrt{\frac{2}{n+\nicefrac{1}{2}}}
			\cos\left(2 b_{n_2,j_2}(i-\nicefrac{1}{2})\right) \\
			&= (-1)^{n_1-j_1}C_2\left(u_{d_2}^{(j_1)}\right)_\ell
		\end{aligned}
      \end{equation*}
		and
		\begin{equation*}
      \begin{aligned}
			\left(v_{d}^{(j)}\right)_{n_1+\ell}
			&= \sqrt{\frac{2}{n+\nicefrac{1}{2}}}
			\sin\left(2b_{n,j}(n_1+\ell-1)\right)\\
			&= \sqrt{\frac{2}{n+\nicefrac{1}{2}}}
			\sin\left( \frac{2\pi (n_1-j_1)n_1}{n_1} + 2b_{n,j}(\ell-1\right) \\
			&= (-1)^{n_1-j_1}\sqrt{\frac{2}{n+\nicefrac{1}{2}}}
			\sin\left(2b_{n_2,j_2}(i-1)\right)\\
			&= (-1)^{n_1-j_1}C_2\left(v_{d_2}^{(j_1)}\right)_\ell.
		\end{aligned}
      \end{equation*}
		\item[Case 3: Open boundary conditions.]
		In this case, the first subregion with $n_1$ compartments has Mixed-Alt boundary conditions (i.e., flux is only allowed at $x=0$) and the second subregion with $n_2$ compartments has Mixed boundary conditions. For $j\in\{1,...,n\}$ suppose there exists a $j_1 \in \{1,...,n_1\}$ such that
		\begin{equation}
      \begin{aligned}
      \label{eq:j1-Dirichlet}
		c_{n,j} &= \frac{\pi(n-j+1)}{2(n+1)} = \frac{\pi(n_1-j_1+\nicefrac{1}{2})}{2n_1+1} = b_{n_1,j_1} \\
		&\implies j_1 = \frac{j(2n_1+1)}{2(n+1)} = C_1^2 j.
      \end{aligned}
		\end{equation}
		From (\ref{app-eq:sigmad}) this implies that $\sigma_{d_1}^{(j_1)}=\sigma^{(j)}$. Let $j_2 = j-j_1 = C_2^2 j$ and note by definition that $j_2 \in \{1,..,n_2\}$. Additionally, $\sigma_{d_2}^{(j_2)} = \sigma^{(j)}$, since
		\begin{equation}
      \begin{aligned}
         \label{eq:j2-Dirichlet}
   		c_{n,j} &= \frac{\pi(n-j+1)}{2(n+1)} \\
         &=	\frac{\pi(2n_2 + 1 - 2 C_2^2j)}{2(2n_2+1)} \\
         &=
   		\frac{\pi(n_2-j_2+\nicefrac{1}{2})}{2n_2+1} = b_{n_2,j_2}.
      \end{aligned}
		\end{equation}

		Next we will prove (\ref{eq:Ud_divided}) and (\ref{eq:Vd_divided}). For $i=1,...,n_1$, (\ref{eq:j1-Dirichlet}) implies that
		\begin{align*}
   		\left(u_{d}^{(j)}\right)_i
   		&= \sqrt{\frac{2}{n+1}}
   		\sin\left(2ic_{n,j}\right) \\
   		&= \sqrt{\frac{2}{n+1}}
   		\sin\left(2ib_{n_1,j_1}\right)
   		= C_1 \left(u_{d_1}^{(j_1)}\right)_i \\
   		\left(v_{d}^{(j)}\right)_i
   		&= \sqrt{\frac{2}{n+1}}
   		\cos\left(2c_{n,j}(i-\nicefrac{1}{2})\right)\\
   		&= \sqrt{\frac{2}{n+1}}
   		\sin\left(2b_{n_1,j_1}(i-\nicefrac{1}{2})\right)
   		= C_1\left(v_{d_1}^{(j_1)}\right)_i.
		\end{align*}
		For $i=n_1+1,...,n$, (\ref{eq:j1-Dirichlet}) and (\ref{eq:j2-Dirichlet}) imply that
		\begin{equation*}
      \begin{aligned}
			\left(u_{d}^{(j)}\right)_{n_1+\ell}
			&= \sqrt{\frac{2}{n+1}}
			\sin\left(2c_{n,j}(n_1+\ell)\right)\\
			&= \sqrt{\frac{2}{n+1}}
			\sin\left(2c_{n,j}(n_1+\nicefrac{1}{2}) + 2c_{n,j}(\ell-\nicefrac{1}{2})\right) \\
			&= \sqrt{\frac{2}{n+1}}
			\sin\left(\frac{2\pi(n_1-j_1+\nicefrac{1}{2})(n_1+\nicefrac{1}{2})}{2n_1+1} + 2c_{n,j}(\ell-\nicefrac{1}{2})\right)\\
			&= (-1)^{n_1-j_1}\sqrt{\frac{2}{n+1}}
			\cos\left(2b_{n_2,j_2}(\ell-\nicefrac{1}{2})\right)\\
			&= (-1)^{n_1-j_1}C_2 \left(u_{d_2}^{(j_2)}\right)_\ell
			\end{aligned}
      \end{equation*}
		and
		\begin{equation*}
      \begin{aligned}
			\left(v_{d}^{(j)}\right)_{n_1+\ell}
			&= \sqrt{\frac{2}{n+1}}
			\cos\left(2 c_{n,j} (n_1+\ell-\nicefrac{1}{2})\right)\\
			&= \sqrt{\frac{2}{n+1}}
			\cos\left(2 c_{n,j}(n_1+\nicefrac{1}{2}) + 2 c_{n,j}(\ell-1)\right) \\
			&= \sqrt{\frac{2}{n+1}}
			\cos\left(\frac{2\pi(n_1-j_1+\nicefrac{1}{2})(n_1+\nicefrac{1}{2})}{2n_1 + 1} + 2c_{n,j}(\ell-1)\right)\\
			&= -(-1)^{n_1-j_1}\sqrt{\frac{2}{n+1}}
			\sin\left(2b_{n_2,j_2}(\ell-1)\right)\\
			&= (-1)^{n_1-j_1} C_2 \left(v_{d_2}^{(j_2)}\right)_\ell.
      \end{aligned}
      \end{equation*}
      which completes the proof.
	\end{description}
\end{proof}

Next, we prove Lemma~\ref{lemma:unpertbEigval}, which provides an eigendecomposition of $S S^T$ when $\gamma = 0$, recall $S$ is given by (\ref{eq:S}). This is equivalent to the non-unique eigendecomposition of the unperturbed matrix $T$, given by (\ref{eq:T}).

\begin{proof}[Proof of Lemma \ref{lemma:unpertbEigval}]
	First note that $T \in \mathbb{R}^{nm\times nm}$ and, as needed, the number of eigenvectors defined is $nm_+ + n_1m_- + n_2m_- = nm$.

	We will show that the matrices $\hat{Q}_{T,1}$, $\hat{Q}_{T,2}$, and $\hat{Q}_{T,3}$ contain eigenvectors of $T$ and $\hat{\Sigma}_{Q_{T,1}}$, $\hat{\Sigma}_{Q_{T,2}}$, and $\hat{\Sigma}_{Q_{T,3}}$ contain the corresponding nonzero eigenvalues. First, considering $\hat{Q}_{T,1}$, we have that
   \begin{equation*}
   \begin{aligned}
   	T\hat{Q}_{T,1}
   	&= \left(S_{d}S_{d}^{T}\otimes D_+ + \left(S_{d}S_{d}^{T}-HH^{T}\right)\otimes D_-\right)(\hat{U}_{d} \otimes I_m^{\mathcal{M}_+}) \\
   	&= S_d S_d^T \hat{U}_{d} \otimes I_m^{\mathcal{M}_+} = \hat{U}_{d} \hat{\Sigma}_d^2 \otimes I_m^{\mathcal{M}_+} = \hat{Q}_{T,1} \left(\hat{\Sigma}_d^2 \otimes I_{m_+}\right)
   \end{aligned}
	\end{equation*}
   Next for $\hat{Q}_{T,2}$, we have that
	\begin{equation*}
   \begin{aligned}
   	T\hat{Q}_{T,2}
      	&= \left(S_{d}S_{d}^{T}\otimes D_+ + \left(S_{d}S_{d}^{T}-HH^{T}\right)\otimes D_-\right)
      	\left(
      	\begin{bmatrix}
      	\hat{U}_{d_1} \\
      	0
      	\end{bmatrix}
      	\otimes I_m^{\mathcal{M}_-}
      	\right) \\
      	&= (S_d S_d^T - H H^T)
      	\begin{bmatrix}
      	\hat{U}_{d_1} \\
      	0
      	\end{bmatrix}
      	\otimes I_m^{\mathcal{M}_-}
      	= \begin{bmatrix}
      	\hat{U}_{d_1} \\
      	0
      	\end{bmatrix} \hat{\Sigma}_{d_1}^2 \otimes I_m^{\mathcal{M}_-}
      	= \hat{Q}_{T,2} \left(\hat{\Sigma}_{d_1}^2\otimes I_{m_-}\right).
   \end{aligned}
   \end{equation*}
	Finally, for $\hat{Q}_{T,3}$ we have that
	\begin{equation*}
   \begin{aligned}
   	T\hat{Q}_{T,3}
   	&= \left(S_{d}S_{d}^{T}\otimes D_+ + \left(S_{d}S_{d}^{T}-HH^{T}\right)\otimes D_-\right)\left(
   	\begin{bmatrix}
   	0 \\
   	\hat{U}_{d_2}
   	\end{bmatrix} \otimes I_m^{\mathcal{M}_-}\right) \\
   	&= (S_d S_d^T - H H^T)\begin{bmatrix}
   	0 \\
   	\hat{U}_{d_2}
   	\end{bmatrix} \otimes I_m^{\mathcal{M}_-}
   	= \hat{Q}_{T,3}\left(\hat{\Sigma}_{d_2}^2\otimes I_{m_-}\right)
   \end{aligned}
	\end{equation*}
	It can analogously be shown that $\breve{Q}_{T,1}$, $\breve{Q}_{T,2}$, and $\breve{Q}_{T,3}$ represent the nullspace of $T$. We leave it as an exercise to show that the eigenvectors and nullspace basis vectors are orthogonal.
\end{proof}

Next, we will show that Theorem~\ref{thm:svd-S} provides an approximate basis for the nullspace of $S$ (i.e., $\breve{V}$), where the basis is orthogonal at small gamma and satisfies $S\breve{V}=0$ in the limit as $\gamma$ goes to zero. We will also provide the proof to Proposition~\ref{prop:breveV}, which gives an exact basis for the nullspace of $S$ that is orthogonal in the limit as $\gamma \rightarrow 0$.

\begin{proof}[Proof of Theorem~\ref{thm:svd-S} (nullspace)]

	The dimension of the nullspace of $S$ is given by Lemma~\ref{lemma:rank}. Notice that this dimension matches the number of columns in $\breve{V}$ as defined in Theorem~\ref{thm:svd-S}. Specifically, for the five matrices that compose $\breve{V}$, i.e. $\breve{V}_i$ for $i=1,...,5$ the number of columns is
   \begin{equation*}
   \begin{aligned}
   	\breve{q}_1 &=
      \begin{cases}
      	\breve{q}_{r_1,-} & \text{Mixed} \\
      	0 & \text{Otherwise}
   	\end{cases} \\
   	\breve{q}_2 &= q_{d_1} p \\
   	\breve{q}_3 &=
         \begin{cases}
         	\breve{q}_{\bar{r}} & \text{Neumann} \\
         	0 & \text{Otherwise}
      	\end{cases} \\
   	\breve{q}_4 &= q_{d_2} p \\
   	\breve{q}_5 &=
         \begin{cases}
         	3m - m_+ & \text{Neumann} \\
            2m - m_+ & \text{Mixed} \\
            m & \text{Open} \\
      	\end{cases} \\
   \end{aligned}
	\end{equation*}
	and by inspection we see that the number of columns is equivlanet to the value of $\breve{q}$ given by Lemma~\ref{lemma:rank}.

We leave it as an exercise to show that all the vectors defined in these matrices are orthonormal.

	To show that the vectors are in the nullspace of $S$, write $S$ as follows
	\begin{equation*}
	S =
	\begin{bmatrix} \vspace{3pt}
	\gamma I_{n_1}\otimes S_{r_1} & 0
	& S_{d,s_1}\otimes D_+ + \begin{bmatrix} S_{d_1} & 0 \end{bmatrix} \otimes D_-\\
	0 & \gamma I_{n_2} \otimes S_{r_2}
	& S_{d,s_2}\otimes D_+ + \begin{bmatrix} 0 & S_{d_2} \end{bmatrix} \otimes D_-
	\end{bmatrix}.
	\end{equation*}
   where $S_{d,s_1}$ represents the first $n_1$ rows of $S_d$ and $S_{d,s_2}$ represents the last $n_2$ rows of $S_d$.
	Suppose a vector in the nullspace can be written as
	\begin{equation*}
	v = \begin{bmatrix}
	v_1 \otimes v_2 \\
	v_3 \otimes v_4 \\
	v_5 \otimes v_6
	\end{bmatrix}\Sigma.
	\end{equation*}
	where $\Sigma$ is a diagonal matrix. Multiplying $S$ by $v$, we obtain the following two equations that must be satisfied
	\begin{align}
	\label{eq:cond1}
	\gamma v_1 \otimes S_{r_1} v_2 + S_{d,s_1} v_5 \otimes D_+ v_6 +
	\begin{bmatrix} S_{d_1} & 0 \end{bmatrix} v_5 \otimes D_- v_6 &= 0 \\
	\label{eq:cond2}
	\gamma v_3 \otimes S_{r_2} v_4 + S_{d,s_2} v_5 \otimes D_+ v_6 +
	\begin{bmatrix} 0 & S_{d_2} \end{bmatrix} v_5 \otimes D_- v_6 &= 0
	\end{align}
	It is straightforward to show that the vectors given by the claim satisfy these equations in the limit as $\gamma \rightarrow 0$. In fact, for $\breve{V}_2$, $\breve{V}_4$ and $\breve{V}_6$ the equations are satisfied at small gamma. Below we will show the logic for $\breve{V}_2$. We leave it as an exercise to verify these results for $\breve{V}_4$ and $\breve{V}_6$. Additionally, it is trivial to show that as $\gamma \rightarrow 0$, $\breve{V}_1$ and $\breve{V}_3$ satisfy the condtions since, in this case, $v_5,v_6=0$.

	For $\breve{V}_2$ we have that
   \begin{equation*}
      \begin{aligned}
         v_1 &= \hat{U}_{d_1} \hat{\Sigma}_{d_1}, & v_2 &= V_{r_1}, &
         v_5 &= \begin{bmatrix}
      	-\gamma \hat{V}_{d_1} \\ 0
      	\end{bmatrix},  & v_6 &= U_{r_1} \Sigma_{r_1}
      \end{aligned}
   \end{equation*}
   ad $v_3,v_4=0$. The first condition, i.e., (\ref{eq:cond1}), is satisfied since
	\begin{equation*}
	\begin{split}
	\gamma \hat{U}_{d_1} \hat{\Sigma}_{d_1} &\otimes S_{r_1} V_{r_1}
	+ S_{d,s_1} \begin{bmatrix}
	-\gamma \hat{V}_{d_1} \\ 0
	\end{bmatrix} \otimes D_+ U_{r_1} \Sigma_{r_1}
	+ \begin{bmatrix} S_{d_1} & 0 \end{bmatrix} \begin{bmatrix}
	-\gamma \hat{V}_{d_1} \\ 0
	\end{bmatrix} \otimes D_- U_{r_1} \Sigma_{r_1} \\
	&= \gamma \hat{U}_{d_1} \hat{\Sigma}_{d_1} \otimes U_{r_1} \Sigma_{r_1}
	- \gamma \hat{U}_{d_1} \hat{\Sigma}_{d_1} \otimes D_+ U_{r_1} \Sigma_{r_1}
	- \gamma \hat{U}_{d_1} \hat{\Sigma}_{d_1} \otimes D_- U_{r_1} \Sigma_{r_1} \\
	&= 0
	\end{split}
	\end{equation*}
	Additionally, $\breve{V}_2$ satisfies (\ref{eq:cond2}) since
	\begin{equation*}
	\begin{split}
	S_{d,s_2} \begin{bmatrix}
	-\gamma \hat{V}_{d_1} \\ 0
	\end{bmatrix} &\otimes D_+ U_{r_1} \Sigma_{r_1} +
	\begin{bmatrix} 0 & S_{d_2} \end{bmatrix} \begin{bmatrix}
	-\gamma \hat{V}_{d_1} \\ 0
	\end{bmatrix} \otimes D_- U_{r_1} \Sigma_{r_1} \\
	&= 0.
	\end{split}
	\end{equation*}

\end{proof}

\begin{proof}[Proof of Proposition~\ref{prop:breveV}]

   The proof to this proposition closely follows the proof given for the nullspace presented in Theorem~\ref{thm:svd-S}. In addition to the logic of this proof we need to show that the basis vectors that differ (i.e., those in $\breve{V}_1$ and $\breve{V}_3$) satisfy the two conditions given in \ref{eq:cond1} and \ref{eq:cond2} at small values of $\gamma$. We will show the logic for $\breve{V}_3$ and leave it as an exercise to show that $\breve{V}_1$ satisfies the conditions.

   For $\breve{V}_3$, when considering the conditions given by \ref{eq:cond1} and \ref{eq:cond2}, we have that
   \begin{equation*}
      \begin{aligned}
         v_1 &= \frac{1}{C_1} \breve{U}_{d,s_1} & v_2 &= \breve{V}_{\bar{r},s_1}, & v_3 &= \frac{1}{C_2}  \breve{U}_{d,s_2} \\ v_4 &= \breve{V}_{\bar{r},s_2}, &
         v_5 &= \gamma w_1, & v_6 &= S_{r_1} \breve{V}_{\bar{r},s_1}
      \end{aligned}
   \end{equation*}
   Note that, by definition of $S_{\bar{r}}$, see (\ref{eq:S_rbar}), the following equations must be satisfied
	\begin{equation}
   \begin{aligned}\label{eq:relations}
   	C_1 S_{r_1,+} \breve{V}_{\bar{r},s_1} &= -C_2 S_{r_2,+} \breve{V}_{\bar{r},s_2} \\
   	S_{r_1,-} \breve{V}_{\bar{r},s_1} &= 0 \\
   	S_{r_2,-} \breve{V}_{\bar{r},s_2} &= 0
   \end{aligned}
	\end{equation}
	From this relations we have that $D_- S_{r_1}\breve{V}_{\bar{r},s_1} = 0$ and $S_{r_1}\breve{V}_{\bar{r},s_1} = D_+ S_{r_1}\breve{V}_{\bar{r},s_1}$.
	It follows that $\breve{V}_3$ satisfies (\ref{eq:cond1}) since
	\begin{equation*}
	\begin{split}
   	\frac{\gamma}{C_1} \breve{U}_{d,s_1} &\otimes S_{r_1} \breve{V}_{\bar{r},s_1}
   	- \gamma S_{d,s_1} w_1 \otimes D_+ S_{r_1}\breve{V}_{\bar{r},s_1}  - \gamma \begin{bmatrix} S_{d_1} & 0 \end{bmatrix} w_1 \otimes D_- S_{r_1}\breve{V}_{\bar{r},s_1}\\
   	&=
      	\frac{\gamma}{C_1} \breve{U}_{d,s_1} \otimes S_{r_1} \breve{V}_{\bar{r},s_1}
      	- \gamma \frac{n}{n_2+\sqrt{n_1 n_2}}S_{d,s_1} \hat{V}_d \hat{\Sigma}_d^{-1} \hat{U}_d^T
      	\begin{bmatrix}
      	\frac{1}{C_1}\breve{U}_{d,s_1} \\
      	-\frac{1}{C_2}\breve{U}_{d,s_2}
      	\end{bmatrix} \otimes S_{r_1} \breve{V}_{\bar{r},s_1} \\
   	&= \frac{\gamma}{C_1} \breve{U}_{d,s_1} \otimes S_{r_1} \breve{V}_{\bar{r},s_1}
      - \gamma \frac{n}{n_2 + \sqrt{n_1 n_2}} \hat{U}_{d,s_1} \hat{U}_d^T
         \begin{bmatrix}
         \frac{1}{C_1}\breve{U}_{d,s_1} \\
         -\frac{1}{C_2}\breve{U}_{d,s_2}
         \end{bmatrix} \otimes S_{r_1} \breve{V}_{\bar{r},s_1} \\
   	&= \frac{\gamma}{C_1} \breve{U}_{d,s_1} \otimes S_{r_1} \breve{V}_{\bar{r},s_1}
      - \frac{\gamma}{C_1} \frac{n_2 + \sqrt{n_1 n_2}}{n_2 + \sqrt{n_1 n_2}}\breve{U}_{d,s_1} \otimes S_{r_1} \breve{V}_{\bar{r},s_1}\\
   	&= 0.
	\end{split}
	\end{equation*}
   Here, we are using the fact that $\breve{U}_d=\frac{1}{\sqrt{n}} \bm{1}$ is a constant vector and
   \begin{equation}
   \begin{aligned}
      & \breve{U}_d\breve{U}_d^T + \hat{U}_d\hat{U}_d^T = I \\
      &\implies \hat{U}_{d,s_1} \hat{U}_{d,s_1}^T = I - \breve{U}_{d,s_1}\breve{U}_{d,s_1}^T \\
      &\implies \hat{U}_{d,s_1} \hat{U}_{d,s_2}^T = - \breve{U}_{d,s_1}\breve{U}_{d,s_2}^T
      \end{aligned}
   \end{equation}therefore
   \begin{equation}
   \begin{aligned}
      \hat{U}_{d,s_1} \hat{U}_d^T
         \begin{bmatrix}
         \frac{1}{C_1}\breve{U}_{d,s_1} \\
         -\frac{1}{C_2}\breve{U}_{d,s_2}
         \end{bmatrix} &=
         \frac{1}{\sqrt{n}}  \begin{bmatrix}
            \hat{U}_{d,s_1} \hat{U}_{d,s_1}^T & \hat{U}_{d,s_1} \hat{U}_{d,s_2}^T
         \end{bmatrix}
         \begin{bmatrix}
         \frac{1}{C_1}\bm{1} \\
         -\frac{1}{C_2}\bm{1}
         \end{bmatrix} \\
      &=  \frac{1}{\sqrt{n}} (
            \frac{1}{C_1}(I - \breve{U}_{d,s_1}\breve{U}_{d,s_1}^T) + \frac{1}{C_2} \breve{U}_{d,s_1}\breve{U}_{d,s_2}^T
         ) \bm{1} \\
      &=  \frac{1}{\sqrt{n}} \left(
            \frac{1}{C_1}(I - a^2 \bm{1}_{n_1\times n_1}) \bm{1}_{n_1} + a^2 \frac{1}{C_2} \bm{1}_{n_1 \times n_2} \bm{1}_{n_2}
         \right) \\
      &=  \frac{1}{\sqrt{n}} \left(
            \frac{1}{C_1}(1 - n_1/n) + \frac{1}{n C_2} n_2
            \right) \bm{1}_{n_1} \\
      &=  \left(
            \frac{1}{C_1}(1 - n_1/n) + \frac{1}{n C_2} n_2
            \right) \breve{U}_{d,s_1}.
   \end{aligned}
   \end{equation}
   Using that, for Neumann boundary conditions, $C_1 = \sqrt{n_1/n}$ and $C_2 = \sqrt{n_2/n}$, we have that
   \begin{equation*}
   \begin{aligned}
      \hat{U}_{d,s_1} \hat{U}_d^T
         \begin{bmatrix}
         \frac{1}{C_1}\breve{U}_{d,s_1} \\
         -\frac{1}{C_2}\breve{U}_{d,s_2}
         \end{bmatrix} &=
         \left(
               \sqrt{\frac{n}{n_1}}(1 - n_1/n) + \sqrt{\frac{1}{n n_2}} n_2
               \right) \breve{U}_{d,s_1} \\
         &=
        \left(
              \sqrt{\frac{n}{n_1}}(1 - n_1/n) + \sqrt{\frac{n_2}{n}}
              \right) \breve{U}_{d,s_1} \\
     &=
    \left(\frac{n_2 + \sqrt{n_1 n_2}}{\sqrt{n_1 n}}
          \right) \breve{U}_{d,s_1} \\
      &=  \frac{1}{C_1}
      \left(\frac{n_2 + \sqrt{n_1 n_2}}{n}
            \right) \breve{U}_{d,s_1}.
   \end{aligned}
   \end{equation*}

   The equalities given by (\ref{eq:relations}) also imply that $S_{r_1} \breve{V}_{\bar{r},s_2}  = -\frac{C_2}{C_1}S_{r_2} \breve{V}_{\bar{r},s_2}$. Using this, we have that,  $\breve{V}_3$ satisfies (\ref{eq:cond2}) since
	\begin{equation*}
	\begin{split}
   	\frac{\gamma}{D_2} \breve{U}_{d,s_2} &\otimes S_{r_2} \breve{V}_{\bar{r},s_2}
   	- \gamma \frac{n}{n_2 + \sqrt{n_1 n_2}} S_{d,s_2} \hat{V}_d \hat{\Sigma}_d^{-1} \hat{U}_d^T
   	\begin{bmatrix}
   	\frac{1}{C_1}\breve{U}_{d,s_1} \\
   	-\frac{1}{C_2}\breve{U}_{d,s_2}
   	\end{bmatrix}  \otimes D_+ S_{r_1} \breve{V}_{\bar{r},s_1} \\
   	&= \frac{\gamma}{C_2} \breve{U}_{d,s_2} \otimes
   	S_{r_2} \breve{V}_{\bar{r},s_2}
   	- \frac{\gamma}{C_2} \frac{n_1 + \sqrt{n_1 n_2}}{n_2 + \sqrt{n_1 n_2}} \frac{C_2}{C_1}\breve{U}_{d,s_2} \otimes
   	S_{r_2} \breve{V}_{\bar{r},s_2} \\
   	&= 0
	\end{split}
	\end{equation*}
	where, similar to the result for the first condition,
   \begin{equation*}
   \begin{aligned}
      \hat{U}_{d,s_2} \hat{U}_d^T
         \begin{bmatrix}
         \frac{1}{C_1}\breve{U}_{d,s_1} \\
         -\frac{1}{C_2}\breve{U}_{d,s_2}
         \end{bmatrix} &=
      &=  \frac{1}{C_2}
      \left(\frac{n_1 + \sqrt{n_1 n_2}}{n}
            \right) \breve{U}_{d,s_1}
   \end{aligned}
   \end{equation*}
   and  $C_2/C_1 = \sqrt{n_2/n_1}$.  We leave it as an exercise to show that all the vectors defined in the columns of $\breve{V}$ are linearly independent and that, in the limit as $\gamma \rightarrow 0$, they become orthogonal.
\end{proof}

\end{document}